\PassOptionsToPackage{x11names}{xcolor}
\PassOptionsToPackage{hypertexnames=false}{hyperref} %

\documentclass[a4paper,11pt,DIV=classic,abstract]{scrartcl}
\usepackage[utf8]{inputenc}
\usepackage[T1]{fontenc}

\usepackage{amsmath,amssymb,amsthm} %
\usepackage{stmaryrd} %
\usepackage{booktabs} %
\usepackage{xcolor}
\usepackage{lmodern} %
\usepackage{float} %
\usepackage{subcaption} %
\usepackage{varwidth} %

\usepackage{mathtools}
\usepackage{thmtools,thm-restate}
\usepackage{csquotes}

\usepackage{algorithm}
\usepackage{algpseudocode}

\usepackage{enumerate}%

\algnewcommand\algorithmicparameter{\textbf{Parameter:}}
\algnewcommand\Parameter{\item[\algorithmicparameter]}
\algnewcommand\algorithmicinput{\textbf{Input:}}
\algnewcommand\Input{\item[\algorithmicinput]}
\algnewcommand\algorithmicoutput{\textbf{Output:}}
\algnewcommand\Output{\item[\algorithmicoutput]}
\algnewcommand\algorithmicoracle{\textbf{Oracle Access:}}
\algnewcommand\Oracle{\item[\algorithmicoracle]}
\algnewcommand{\algrule}{\par\vskip.5\baselineskip\hrule height .4pt\par\vskip.5\baselineskip}
\algnewcommand{\iIf}[1]{\State\algorithmicif\ #1\ \algorithmicthen}
\algnewcommand{\iEndIf}{\unskip\ \algorithmicend\ \algorithmicif}
\algnewcommand{\iWhile}[1]{\State\algorithmicwhile\ #1\ \algorithmicdo}
\algnewcommand{\iEndWhile}{\unskip\ \algorithmicend\ \algorithmicwhile}
\algnewcommand{\iFor}[1]{\State\algorithmicfor\ #1\ \algorithmicdo}
\algnewcommand{\iEndFor}{\unskip\ \algorithmicend\ \algorithmicfor}

\let\from\colon
\let\with\colon

\let\epsilon\varepsilon
\let\phi\varphi

\DeclarePairedDelimiter{\set}{\lbrace}{\rbrace}
\DeclarePairedDelimiterX{\bagelem}[2]{\langle}{\rangle}{#1,#2}

\DeclareMathOperator{\ar}{ar}

\DeclareMathOperator{\lc}{lc}

\DeclarePairedDelimiter{\sem}{\llbracket}{\rrbracket}
\newcommand*{\?}{\:\cdot\:}

\DeclareMathOperator{\adom}{adom}

\newcommand*{\NN}{\mathbb N}

\newcommand*{\RR}{\mathbb R}

\newcommand*{\EXPECTATION}{\mathsf{EXPECTATION}}
\newcommand*{\VARIANCE}{\mathsf{VARIANCE}}
\newcommand*{\PQE}{\mathsf{PQE}}
\newcommand*{\PQEset}{\PQE\sp{\mathsf{set}}}
\newcommand*{\sharpSUBSETSUM}{\sharp\mathsf{SUBSETSUM}}

\makeatletter
\newcommand*{\@complexitystyle}[1]{\mathsf{#1}}
\newcommand*{\sharpP}{\sharp\@complexitystyle{P}}
\newcommand*{\FP}{\@complexitystyle{FP}}
\newcommand*{\PTIME}{\@complexitystyle{PTIME}}
\makeatother

\DeclareMathOperator*{\E}{E}
\DeclareMathOperator*{\Var}{Var}
\DeclareMathOperator*{\Cov}{Cov}

\newcommand*{\rep}{\mathrm{Rep}}
\newcommand*{\tabs}{\mathbf{T}}

\newcommand*{\true}{\mathsf{true}}
\newcommand*{\false}{\mathsf{false}}

\newcommand*{\DD}{\mathbb{D}}

\newcommand*{\D}{\ensuremath{\mathcal{D}}}

\DeclarePairedDelimiter{\cod}{\langle}{\rangle}

\DeclarePairedDelimiter{\size}{\lvert}{\rvert}

\DeclareMathOperator{\sg}{at}

\newcommand*{\spacestyle}[1]{\mathbb{#1}}

\newcommand*{\Univ}{\mathrm{dom}}
\newcommand*{\Vars}{\spacestyle{V}}

\renewcommand*{\#}{\ensuremath{\sharp}}

\newcommand{\tup}[1]{\boldsymbol{\mathbf{#1}}}

\newcommand*{\inflate}{\mathsf{inflate}}
\newcommand*{\solveComponent}{\mathsf{solveComponent}}

\DeclarePairedDelimiterX{\dbraces}[1]{\lbrace}{\rbrace}{%
	\nbrace{\lbrace}{#1}\delimsize\lbrace\mathopen{}%
	#1%
	\mathclose{}\delimsize\rbrace\nbrace{\rbrace}{#1}%
}
\newcommand{\dummydelim}[2]{$\left#1\vphantom{#2}\right.$}
\newcommand{\nbrace}[2]{\sbox0{\dummydelim{#1}{#2}}\hspace{\the\dimexpr -0.85\wd0 + 2pt\relax}}

\makeatletter
\def\bag{\@ifstar\@bag\@@bag}
\def\@bag#1{\dbraces{\smash{#1}}}
\def\@@bag#1{\dbraces*{#1}}
\newcommand{\Bags}[3][\@nil]{%
	\def\tmp{#1}%
	\ifx\tmp\@@nil%
		\dparens*{\smash{\begin{smallmatrix}#2\\#3\end{smallmatrix}}}%
	\else%
		\dparens*{\begin{smallmatrix}#2\\#3\end{smallmatrix}}%
	\fi%
}
\makeatother

\newcommand{\paramproblem}[4]{%
    \begin{center}
        \vspace{\baselineskip}%
        \begin{tabular}{p{.275\linewidth} p{.6\linewidth} }
            \multicolumn{2}{l}{\textbf{Problem} \enspace #1}      \\\toprule
            \textsc{Parameter:} &   #2  \\
            \textsc{Input:}     &   #3  \\
            \textsc{Output:}    &   #4  \\\bottomrule
        \end{tabular}%
        \vspace{\baselineskip}
    \end{center}
}

\usepackage{hyperref} %
\definecolor{myblue}{HTML}{3D6288}
\definecolor{myred}{HTML}{9F1D2B}
\definecolor{mygreen}{HTML}{5B892D}
\hypersetup{
    breaklinks,
    pdfencoding=auto,
    colorlinks,
    linkcolor=myred,%
    citecolor=mygreen,%
    urlcolor=myblue,%
}
\usepackage[capitalise]{cleveref}

\newtheorem{theorem}{Theorem}[section]
\newtheorem{lemma}[theorem]{Lemma}
\newtheorem{proposition}[theorem]{Proposition}
\newtheorem{corollary}[theorem]{Corollary}
\theoremstyle{definition}

\theoremstyle{remark}
\newtheorem{example}[theorem]{Example}
\newtheorem{remark}[theorem]{Remark}

\newenvironment{sketch}{\begin{proof}[Proof Sketch]}{\end{proof}}

\theoremstyle{definition}
\newtheorem{definition2}[theorem]{Definition}

\DeclareMathOperator{\Bern}{Bernoulli}
\DeclareMathOperator{\Binom}{Binomial}
\DeclareMathOperator{\Geom}{Geometric}

\newcommand*{\strBern} {\mathtt{Bernoulli}}
\newcommand*{\strBinom}{\mathtt{Binomial}}
\newcommand*{\strGeom} {\mathtt{Geometric}}
\newcommand*{\strPoiss}{\mathtt{Poisson}}

\newcommand*{\probs}{\mathbf P}
\DeclareMathOperator{\zeroPr}{\mathsf{zeroPr}}

\usepackage{tikz}
\usetikzlibrary{matrix}
\usetikzlibrary{fit}
\usetikzlibrary{shapes}
\usetikzlibrary{backgrounds}
\usetikzlibrary{arrows}

\DeclareMathOperator{\Count}{Count}

\DeclareMathOperator{\Sum}{Sum} \title{{\LARGE{}Probabilistic Query Evaluation with Bag Semantics}}
\usepackage{authblk}
\usepackage{libertine}
\usepackage{libertinust1math}
\usepackage[scaled=0.83]{beramono}
\usepackage[T1]{fontenc}

\usepackage{varwidth}

\author[1]{Martin Grohe}
\author[2]{Peter Lindner}
\author[1]{Christoph Standke}
\affil[1]{\textit{\{grohe,standke\}@informatik.rwth-aachen.de},\protect\\RWTH Aachen University, Aachen, Germany}
\affil[2]{\textit{peter.lindner@epfl.ch},\protect\\École Polytechnique Fédérale de Lausanne (EPFL), Lausanne, Switzerland}

\date{\large{}\today}

\usepackage{array}

\begin{document}

\allowdisplaybreaks

\maketitle %

\begin{abstract}
        We study the complexity of evaluating queries on probabilistic databases under bag semantics.
    We focus on self-join free conjunctive queries, and probabilistic databases where occurrences of different facts are independent, which is the natural generalization of tuple-independent probabilistic databases to the bag semantics setting. For set semantics, the data complexity of this problem is well understood, even for the more general class of unions of conjunctive queries: it is either in polynomial time, or $\sharpP$-hard, depending on the query (Dalvi \&{} Suciu, JACM 2012).

    A reasonably general model of bag probabilistic databases may have unbounded multiplicities. In this case, the probabilistic database is no longer finite, and a careful treatment of representation mechanisms is required.
    Moreover, the answer to a Boolean query is a probability distribution over %
    (possibly \emph{all})
    non-negative integers, rather than a probability distribution over $\set{ \true, \false }$. Therefore, we discuss two flavors of probabilistic query evaluation: computing expectations of answer tuple multiplicities, and computing the probability that a tuple is contained in the answer at most $k$ times for some parameter $k$. Subject to mild technical assumptions on the representation systems, it turns out that expectations are easy to compute, even for unions of conjunctive queries. For query answer probabilities, we obtain a dichotomy between solvability in polynomial time and $\sharpP$-hardness for self-join free conjunctive queries.
\end{abstract}

\section{Introduction}\label{sec:intro}

Probabilistic databases (PDBs) provide a framework for managing uncertain data. In database theory, they have been intensely studied since the late 1990s \cite{Suciu+2011,VandenBroeckSuciu2017}.
Most efforts have been directed towards tuple-independent relational databases \emph{under a set semantics}. 
Many relational database systems, however, use a bag semantics, where identical tuples may appear several times in the same relation. Despite receiving little attention so far, bag semantics are also a natural setting for probabilistic databases. For example, they naturally enter the picture when aggregation is performed, or when statistics are computed (e.g., by random sampling, say, without replacement). Either case might involve computing projections without duplicate elimination first. Even when starting from a tuple-independent probabilistic database with set semantics, this typically gives rise to (proper) bags. %
Even in the traditional setting of a PDB where only finitely many facts appear with non-zero probability, under a bag semantics we have to consider infinite probability spaces 
\cite{GroheLindner2022,GroheLindner2022a},
simply because there is no a priori bound on the number of times a fact may appear in a bag. In general, while the complexity landscape of query answering is well understood for simple models of PDBs under set semantics, the picture for \emph{bag semantics} is still 
mostly
unexplored.

Formally, probabilistic databases are probability distributions over conventional database instances. In a database instance, the answer to a Boolean query under \emph{set semantics} is either $\true$ ($1$) or $\false$ ($0$). In a probabilistic database, the answer to such a query becomes a $\set{0,1}$-valued random variable. The problem of interest is \emph{probabilistic query evaluation}, that is, computing the probability that a Boolean query returns $\true$, when given a probabilistic database. The restriction to \emph{Boolean} queries comes with no loss of generality: to compute the probability of any tuple in the result of a non-Boolean query, all we have to do is replace the free variables of the query according to the target tuple, and solve the problem for the resulting Boolean query \cite{Suciu+2011}.

Under a bag semantics, a Boolean query is still just a query without free variables, but the answer to Boolean query can be any non-negative integer, which can be interpreted as the multiplicity of the empty tuple in the query answer, or more intuitively as the number of different ways in which the query is satisfied. In probabilistic query evaluation, we then get $\NN$-valued answer random variables. Still, the reduction from the non-Boolean to the Boolean case works as described above. \emph{Therefore, without loss of generality, we only discuss Boolean queries in this paper.}

As most of the database theory literature, we study the \emph{data complexity} of query evaluation \cite{Vardi1982}, that is, the complexity of the problem, when the query $Q$ is fixed, and the PDB is the input.
The standard model for complexity theoretic investigations is that of \emph{tuple-independent} PDBs, where the distinct facts constitute independent events. Probabilistic query evaluation is well-understood for the class of \emph{unions of conjunctive queries} (UCQs) on PDBs that are tuple-independent (see the related works section below). 
Most prior work, however, considers the problem under plain set semantics or in finite and restricted settings.
Here, on the contrary, we discuss the probabilistic query evaluation under \emph{bag semantics}.

For tuple-independent (set) PDBs, a variety of representation systems have been proposed (cf.\  \cite{GreenTannen2006,Suciu+2011}), although for complexity theoretic discussions, it is usually assumed that the input is just given as a table of facts, together with their marginal probabilities \cite{VandenBroeckSuciu2017}. In the bag version of tuple-independent PDBs \cite{GroheLindner2022a}, different facts are still independent. Yet, the individual facts (or, rather, their multiplicities) are $\NN$-valued, instead of Boolean, random variables. As this, in general, rules out the naive representation through a list of facts, multiplicities, and probabilities, it is necessary to first define suitable representation systems before the complexity of computational problems can be discussed.

Once we have settled on a suitable class of representations, we investigate the problem of probabilistic query evaluation again, subject to representation system $\rep$. Under bag semantics, there are now \emph{two} natural computational problems regarding query evaluation: $\EXPECTATION_{\rep}(Q)$, which is computing the expected outcome, and $\PQE_{\rep}(Q,k)$ which is computing the probability that the outcome is at most $k$. Notably, these two problems coincide for set semantics, because the expected value of a $\{0,1\}$-valued random variables coincides with the probability that the outcome is $1$. Under a bag semantics, however, the two versions exhibit quite different properties 
(cf.~\cite{ReSuciu2009,fink2012aggregation}).

Recall that using a set semantics, unions of conjunctive queries can either be answered in polynomial time, or are $\sharpP$-hard \cite{DalviSuciu2012}. Interestingly, computing expectations using a bag semantics is extraordinarily easy in comparison: with only mild assumptions on the representation, the expectation of any UCQ can be computed in polynomial time. 
Furthermore, the variance of the random variable can also be computed in polynomial time, which via Chebyshev's inequality gives us a way to estimate the probability that the query answer is close to its expectation.
These results contrast the usual landscape of computational problems in uncertain data management, which are rarely solvable efficiently. 

The computation of probabilities of concrete answer multiplicities, however, appears to be less accessible, and in fact, in its properties is more similar to the set semantics version of probabilistic query evaluation. Our main result states that for Boolean conjunctive queries without self-joins, we have a dichotomy between polynomial time and $\sharpP$-hardness of the query. This holds whenever efficient access to fact probabilities is guaranteed by the representation system and is independent of $k$. Although the proof builds upon ideas and notions introduced for the set semantics dichotomy \cite{DalviSuciu2004,DalviSuciu2007a,DalviSuciu2012}, we are confronted with a number of completely new and intricate technical challenges due to the change of semantics. 
On the one hand, the bag semantics turns disjunctions and existential quantification into sums. This facilitates the computation of expected values, because it allows us to exploit linearity.
On the other hand, the new semantics 
(and the potential presence of infinite multiplicity distributions)
keep us from directly applying some of the central ideas from \cite{DalviSuciu2012} when analyzing $\PQE_{\rep}(Q,k)$, thus necessitating novel techniques.
The bag semantics dichotomy for answer count probabilities is, hence, far from being a simple corollary from the set semantics dichotomy. From the technical perspective, the most interesting result is the transfer of hardness from $\PQE_{\rep}(Q,0)$ to $\PQE_{\rep}(Q,k)$. In essence, we need to find a way to compute the probability that $Q$ has $0$ answers, with only having access to the probability that $Q$ has at most $k$ answers for any single fixed $k$. This reduction uses new non-trivial techniques: by manipulating the input table, we can construct multiple instances of the $\PQE_{\rep}(Q,k)$ problem. We then transform the solutions to these problems, which are obtained through oracle calls, into function values of a polynomial (with a priori unknown coefficients) in such a way, that the solution to $\PQE_{\rep}(Q,0)$ on the original input is hidden in the leading coefficient of this polynomial. 
Using a technique from polynomial interpolation, we can find these leading coefficients, and hence, solve $\PQE_{\rep}(Q,0)$.
\paragraph*{Related Work}\label{sec:relatedwork}

The most prominent result regarding probabilistic query evaluation is the Dichotomy theorem by Dalvi and Suciu \cite{DalviSuciu2012} that provides a separation between unions of conjunctive queries for which probabilistic evaluation is possible in polynomial time, and such where the problem becomes $\sharpP$-hard. They started their investigations with self-join free conjunctive queries \cite{DalviSuciu2007a} and later extended their results to general CQs \cite{DalviSuciu2007b} and then UCQs \cite{DalviSuciu2012}. Beyond the queries they investigate, there are a few similar results for fragments with negations or inequalities \cite{FinkOlteanu2016,OlteanuHuang2008,OlteanuHuang2009}, for homomorphism-closed queries \cite{AmarilliCeylan2020}%
, and on restricted classes of PDBs \cite{Amarilli+2016}. Good overviews over related results are given in \cite{VandenBroeckSuciu2017,Suciu2020}. In recent developments, the original dichotomies for self-join free CQs, and for general UCQs have been shown to hold even under severe restrictions to the fact probabilities that are allowed to appear \cite{AmarilliKimelfeld2021,KenigSuciu2021}.

The bag semantics for CQs we use here is introduced in \cite{ChaudhuriVardi1993}. A detailed analysis of the interplay of bag and set semantics is presented in \cite{Cohen2009}. Considering multiplicities as semi-ring annotations \cite{Green+2007,GraedelTannen2017}, embeds bag semantics into a broader mathematical framework. 

In two closely related papers \cite{ReSuciu2009,fink2012aggregation}, the authors study various aggregates (including, in particular, $\Count$ and $\Sum$) over select-project-join queries on semi-ring annotated tuple-independent PDBs. This has direct implications for query evaluation with bag semantics and, in particular, implies some of our results, at least for finite probability distributions. Specifically, for full count aggregations, this is equivalent to the semantics of Boolean CQs that are discussed in \cref{sec:prelim}. In this sense, these papers discuss a variant of our $\PQE(Q,k)$ problem.
Neither of the papers discusses the impact of the representation of probability distributions and restrictions thereof in greater detail. We explain the differences, and the contributions with respect to $\PQE(Q,k)$ below and have added some additional remarks within the main part of our paper.
In the work of Ré and Suciu \cite{ReSuciu2009}, the input is just a tuple-independent PDB, and the annotations for answering $\Count$ aggregate queries are fixed to be $1$ for all tuples, in the semi-ring of integers modulo $k+1$. The value $k$ itself is part of the input and assumed to be given in binary encoding (rather than constant, as in our later discussion). In general, they do not allow arbitrary semi-ring annotations, but have fixed semi-ring annotations for a given tuple-independent PDB $\D$, depending on $\D$ and the query $Q$. In this setting, they establish a dichotomy in \cite[Theorem 2]{ReSuciu2009}: If the Boolean CQ $Q$ is non-hierarchical, then computing answer count probabilities is $\sharpP$-hard, otherwise, it can be solved in polynomial time (the latter being extended by the results of \cite{fink2012aggregation}, see below). They do not give the details for hardness of $\Count$, but the argument can be assumed to work similar to our proof of \cref{pro:hardlambdas} for $k=0$, and does not easily extend to arbitrary constant $k$ in our more general setting (see \cref{thm:zero-to-k-reduction} and our surrounding discussion). In conclusion, concerning hardness, their techniques only suffice to directly infer $\sharpP$-hardness of our $\PQE(Q,k)$ problem for non-hierarchical Boolean CQs $Q$, $k=0$, and their special, restricted class of representations. They do cover more general annotations for $\Sum$ queries, but there is no direct way to transfer these results back to our $\PQE(Q,k)$ problem.

Fink, Han and Olteanu \cite{fink2012aggregation} use a highly more general representation system, called \emph{pvc-tables}, in which tuples in the input tuple-independent PDB may have values and annotations that are general semi-ring or semi-module expressions in independent random variables. In contrast, our representation system in \cref{sec:rep} would correspond to annotating each tuple with a distinct ($\NN$-valued) random variable. They establish tractability for a class of queries covering self-join free, hierarchical Boolean CQs by compiling queries into $d$-trees. Such a $d$-tree encodes the probability distribution of answer counts in terms of the probability distributions of random variables in the annotations. They suggest to support distributions with infinite, even uncountable support (pointing to \cite{kennedy2010pip}), but do not discuss their encoding and its impact 
on tractability.
In particular, their statements about evaluating compiled $d$-trees \cite[Remark 1 \&{}~Theorem 2]{fink2012aggregation} assume finite distribution supports. 
Moreover, to obtain their tractability result for $\Count$, they go back to more restricted annotations similar to \cite{ReSuciu2009}, see \cite[Proposition 3]{fink2012aggregation}. 
Still,
they essentially cover \cref{thm:PQEhiersjfCQeasy}, and the algorithm we present in \cref{app:pqe} can be seen as a special case of theirs. %
The same algorithm is underlying \cite[Theorem 1]{ReSuciu2009}. 
Fink et al.
do not discuss hardness beyond what is said in \cite{ReSuciu2009}.
Finally, in recent work (independent of ours), Feng~et~al.~\cite{Feng+2022} analyze the fine-grained complexity of computing expectations of queries in probabilistic bag databases, albeit assuming finite multiplicity supports and hence still in the realm of finite probabilistic databases.

\section{Preliminaries}
\label{sec:prelim}
We denote by $\NN$ and $\NN_+$ the sets of non-negative, and of positive integers, respectively. We denote open, closed and half open intervals of real numbers by $(a,b)$, $[a,b]$, $[a,b)$ and $(a,b]$, respectively, where $a\leq b$. %
By $\binom{n}{k}$ we denote the binomial coefficient and by $\binom{ n }{ n_1, \dots, n_k }$ the multinomial coefficient.

Let $\Omega$ be a non-empty finite or countably infinite set and let $P \from \Omega \to [0,1]$ be a function satisfying $\sum_{ \omega \in \Omega } P(\omega) = 1$. Then $(\Omega,P)$ is a \emph{(discrete) probability space}. Subsets $A \subseteq \Omega$ are called \emph{events}. We write $\Pr_{ \omega \sim \Omega }( \omega \in A )$ for the probability of a randomly drawn $\omega \in \Omega$ (distributed according to $P$) to be in $A$. More generally, we may write $\Pr_{ \omega \in \Omega }( \omega \text{ has property } \phi )$ for the probability of a randomly drawn element to satisfy some property $\phi$. \emph{All probability spaces appearing in this paper are discrete.}

Functions $X \from \Omega \to \RR$ on a probability space are called \emph{random variables}. The expected value and variance of $X$ are denoted by $\E(X)$ and $\Var(X)$, respectively. The values $\E(X^k)$ for integers $k \geq 2$ are called the higher-order \emph{moments} of $X$.

\subsection{Probabilistic Bag Databases}
We fix a countable, non-empty set $\Univ$ (the \emph{domain}). 
A \emph{database schema} $\tau$ is a finite, non-empty set of relation symbols. Every relation symbol $R$ has an arity $\ar(R) \in \NN_+$.

A \emph{fact} over $\tau$ and $\Univ$ is an expression $R( \tup a )$ where $\tup a \in \Univ^{\ar(R)}$. A \emph{(bag) database instance} $D$ is a bag (i.e. multiset) of facts. Formally, a bag (instance) is specified by a function $\#_D$ that maps every fact $f$ to its multiplicity $\#_D( f )$ in $D$. The \emph{active domain} $\adom(D)$ is the set of domain elements $a$ from $\Univ$ for which there exists a fact $f$ containing $a$ such that $\#_D( f ) > 0$.

A \emph{probabilistic (bag) database} (or, \emph{(bag) PDB}) $\D$ is a pair $(\DD, P)$ where $\DD$ is a set of bag instances and $P \from 2^{\DD} \to [0,1]$ is a probability distribution over $\DD$. Note that, even when the total number of different facts is finite, $\DD$ may be infinite, as facts may have arbitrarily large multiplicities. We let $\#_{\D}( f )$ denote the random variable $D \mapsto \#_D( f )$ for all facts $f$. If $\D = (\DD,P)$ is a PDB, then $\adom( \D ) \coloneqq \bigcup_{ D \in \DD } \adom( D )$. We call a PDB \emph{fact-finite} if the set $\set{ f \with \#_D(f) > 0 \text{ for some } D \in \DD }$ is finite. In this case, $\adom( \D )$ is finite, too.

A bag PDB $\D$ is called \emph{tuple-independent} if for all $k \in \NN$, all pairwise distinct facts $f_1, \dots, f_k$, and all $n_1, \dots, n_k \in \NN$, the events $\#_D( f_i ) = n_i$ are independent, i.\,e.,
\[
    \Pr_{ D \sim \D }\big( \#_D(f_i) = n_i \text{ for all } i = 1, \dots, k \big)
        =
    \prod_{ i = 1 }^{ k } \Pr_{ D \sim \D }\big( \#_D(f_i) = n_i \big)\text.
\]
\emph{Unless it is stated otherwise, all probabilistic databases we treat in this paper are assumed to be fact-finite and tuple-independent.}

\subsection{UCQs with Bag Semantics}
Let $\Vars$ be a countably infinite set of variables. An \emph{atom} is an expression of the shape $R(\tup t)$ where $R \in \tau$ and $\tup t \in ( \Univ \cup \Vars )^{\ar(R)}$. 
A \emph{conjunctive query (CQ)} is a formula $Q$ of first-order logic (over $\tau$ and $\Univ$) of the shape
\[
    Q 
        = 
    \exists x_1 \dots \exists x_m \with 
    R_1( \tup t_1 ) \wedge \dots \wedge R_n( \tup t_n )\text,
\]
in which we always assume that the $x_i$ are pairwise different, and that $x_i$ appears in at least one of $\tup t_1,\dots,\tup t_n$ for all $i = 1,\dots, m$. 
A CQ $Q$ is \emph{self-join free}, if every relation symbol occurs at most once within $Q$. In general, the \emph{self-join width} of a CQ $Q$ is the maximum number of repetitions of the same relation symbol in $Q$.
If $Q$ is a CQ of the above shape, we let $Q^*$ denote the quantifier-free part $R_1( \tup t_1) \wedge \dots R_n( \tup t_n )$ of $Q$, and we call $R_i( \tup t_i )$ an \emph{atom of $Q$} for all $i = 1, \dots, n $.
A \emph{union of conjunctive queries (UCQ)} is a formula of the shape
$
    Q
        =
    Q_1 \vee \dots \vee Q_N
$ 
where $Q_1,\dots,Q_N$ are CQs.
A query is called \emph{Boolean}, if it contains no free variables (that is, there are no occurrences of variables that are not bound by a quantifier). \emph{From now on, and throughout the remainder of the paper, we only discuss Boolean (U)CQs.} 

\medskip

Recall that $\#_D$ is the multiplicity function of the instance $D$. The bag semantics of (U)CQs extends $\#_D$ to queries. For Boolean CQs $Q = \exists x_1 \dots \exists x_m \with R_1( \tup t_1 ) \wedge \dots \wedge R_n( \tup t_n )$ we define
\begin{equation}\label{eq:CQsem}
    \#_D( Q ) 
        \coloneqq
    \sum_{ \tup a \in \adom(D)^m }
        \prod_{ i = 1 }^{ n }
        \#_D\big( R_i( \tup t_i[ \tup x / \tup a ] ) \big)\text,
\end{equation}
where $\tup x = (x_1,\dots,x_m)$ and $\tup a = (a_1,\dots,a_m)$, and $R_i(\tup t_i[ \tup x / \tup a ])$ denotes the fact obtained from $R_i(\tup t_i)$ by replacing, for all $j = 1, \dots, m$, every occurrence of $x_j$ by $a_j$. If $Q = Q_1 \vee \dots \vee Q_N$ is a Boolean UCQ, then each of the $Q_i$ is a Boolean CQ. We define
\begin{equation}\label{eq:UCQsem}
    \#_D( Q )
        \coloneqq
    \#_D( Q_1 ) + \dots + \#_D( Q_N )\text.
\end{equation}
Whenever convenient, we write $\#_D Q$ instead of $\#_D(Q)$. We emphasize once more, that the query being Boolean \emph{does not} mean that its answer is $0$ or $1$ under bag semantics, but could be any non-negative integer.

\begin{remark}
    We point out that in \eqref{eq:CQsem}, conjunctions should intuitively be understood as joins rather than intersections. Our definition \eqref{eq:CQsem} for the bag semantics of CQs matches the one that was given in \cite{ChaudhuriVardi1993}. This, and the extension \eqref{eq:UCQsem} for UCQs, are essentially special cases of how semiring annotations of formulae are introduced in the provenance semiring framework \cite{Green+2007,GraedelTannen2017}, the only difference being that we use the active domain semantics. For UCQs however, this is equivalent since the value of \eqref{eq:CQsem} stays the same when the quantifiers range over arbitrary supersets of $\adom(D)$.
\end{remark}

Note that the result $\#_D Q$ of a Boolean UCQ on a bag instance $D$ is a non-negative integer. Thus, evaluated over a PDB $\D = (\DD, P)$, this yields a $\NN$-valued random variable $\#_{\D} Q$ with
\[
    \Pr\big( \#_{\D} Q = k \big)
        =
    \Pr_{ D \sim \D }\big( \#_D Q = k \big)\text.
\]

\begin{example}
    Consider tuple-independent bag PDB over facts $R(a)$ and $S(a)$, where $R(a)$ has multiplicity $2$ or $3$, both with probability $\frac12$, and $S(a)$ has multiplicity $1$, $2$ or $3$, with probability $\frac13$ each. Then, the probability of the event $\#_{\D}( R(a) \wedge S(a) ) = 6$ is given by
    \begin{math}
        \Pr\big( \#_{\D}( R(a) ) = 2 \big)  \Pr\big( \#_{\D}( S(a) ) = 3 \big)
        + \Pr\big( \#_{\D}( R(a) ) = 3 \big)  \Pr\big( \#_{\D}( S(a) ) = 2 \big)
        = \tfrac13
        \text.\qedhere
    \end{math}
\end{example}

There are now two straight-forward ways to formulate the problem of answering a Boolean UCQ over a probabilistic database. We could either ask for the expectation $\E\big( \#_D Q \big)$, or compute the probability that $\#_{\D} Q$ is at most / at least / equal to $k$. These two options coincide for set semantics, as $\#_{\D} Q$ is $\set{0,1}$-valued in this setting.\footnote{In fact, in the literature both approaches have been used to introduce the problem of probabilistic query evaluation \cite{Suciu+2011,VandenBroeckSuciu2017}.} For bag PDBs, these are two separate problems to explore. Complexity-wise, we focus on \emph{data complexity} \cite{Vardi1982}. That is, the query (and for the second option, additionally the number $k$) is a parameter of the problem, so that the input is only the PDB. Before we can start working on these problems, we first need to discuss how bag PDBs are presented as an input to an algorithm. This is the purpose of the next section. %
\section{Representation Systems}
\label{sec:rep}
For the set version of probabilistic query evaluation, the default representation system represents tuple-independent PDBs by specifying all facts together with their marginal probability. The distinction between a PDB and its representation is then usually blurred in the literature. This does not easily extend to bag PDBs, as the distributions of $\#_D( f )$ for facts $f$ may have infinite support.

\begin{example}\label{exa:geometric}
Let $\D = ( \DD, P )$ be a bag PDB over a single fact $f$ with multiplicity distribution $\#_{\D}( f ) \sim \Geom\big( \frac12 \big)$, i.e., $\Pr_{ D \sim \D }\big( \#_D(f) = k \big) = 2^{-k}$. Then the instances of $\D$ with positive probability are $\bag{}, \bag{f}, \bag{f,f}, \dots$, so $\D$ is an infinite PDB.
\end{example}

To use such PDBs as inputs for algorithms, we introduce a suitable class of representation systems (RS) \cite{GreenTannen2006}. All computational problems are then stated with respect to an RS.

\begin{definition2}[cf.\ \cite{GreenTannen2006}]\label{def:repsys}
    A \emph{representation system (RS)} for bag PDBs is a pair $\big( \tabs, \sem{\?} \big)$ where $\tabs$ is a non-empty set (the elements of which we call \emph{tables}), and $\sem{\?}$ is a function that maps every $T \in \tabs$ to a probabilistic database $\sem{T}$.
\end{definition2}

Given an RS, we abuse notation and also use $T$ to refer to the PDB $\sem{T}$. Note that \cref{def:repsys} is not tailored to tuple-independence yet and requires no independence assumptions. For representing tuple-independent bag PDBs, we introduce a particular subclass of RS's where facts are labeled with the parameters of parameterized distributions over multiplicities. For example, a fact $f$ whose multiplicity is geometrically distributed with parameter $\frac{1}{2}$ could be annotated with $(\strGeom,\mathtt{1/2})$, representing $\frac{1}{2}$ using two integers.

\begin{definition2}\label{def:paramtirep}
    A \emph{parameterized TI representation system} (in short: TIRS) is a tuple $\rep = ( \Lambda, \probs, \Sigma, \tabs, \cod{\?}, \sem{\?} )$ where $\Lambda \neq \emptyset$ is a set (the \emph{parameter set}); $\probs$ is a family $\big( P_\lambda \big)_{\lambda \in \Lambda}$ of probability distributions $P_\lambda$ over $\NN$; 
    $\Sigma \neq \emptyset$ is a finite set of symbols (the \emph{encoding alphabet}); $\cod{ \? } \from \Lambda \to \Sigma^*$ is an injective function (the \emph{encoding function}); and $( \tabs, \sem{ \? } )$ is an RS where 
        \begin{itemize}
            \item $\tabs$ is the family of all finite sets $T$ of pairs $\big( f, \cod{ \lambda_f } \big)$ with pairwise different facts $f$ of a given schema and $\lambda_f \in \Lambda$ for all $f$; and
            \item $\sem{\?}$ maps every $T \in \tabs$ to the tuple-independent bag PDB $\D$ with multiplicity probabilities $\Pr\bigl( \#_{\D} f = k \bigr) = P_{\lambda_f}(k)$ for all $\bigl( f, \cod{\lambda_f} \bigr) \in T$.\qedhere
        \end{itemize}
\end{definition2}

Whenever a TIRS $\rep$ is given, we assume $\rep = ( \Lambda_{\rep}, \probs_{\rep}, \Sigma_{\rep}, \tabs_{\rep}, \cod{ \? }_{\rep},$ $\sem{ \? }_{\rep} )$ by default. %

\begin{figure}
    \hfill%
    \begin{tabular}[t]{ c l }\toprule
        Relation $R$ & Parameter\\\midrule
        $R(1,1)$ & $(\strBern,\mathtt{1/2})$\\
        $R(1,2)$ & $(\strBinom,\mathtt{10},\mathtt{1/3})$\\
        $R(2,2)$ & $(\mathtt{0} \mapsto \mathtt{1/4}; \mathtt{1}\mapsto \mathtt{1/4}; \mathtt{5}\mapsto \mathtt{1/2})$\\\bottomrule
    \end{tabular}%
    \hfill%
    \begin{tabular}[t]{ c l } \toprule
        Relation $S$ & Parameter\\\midrule
        $S(1)$ & $(\strGeom,\mathtt{1/3})$\\
        $S(2)$ & $(\strPoiss,\mathtt{3})$\\\bottomrule
    \end{tabular}%
    \hfill\mbox{}%
    \caption{Example of a parameterized TI representation.}
    \label{fig:repexample}
\end{figure}

\begin{example}\label{exa:dists}
    \Cref{fig:repexample} shows a table $T$ from a TIRS $\rep$, illustrating how the parameters can be used to encode several multiplicity distributions. Four of the distributions are standard parameterized distributions, presented using their symbolic name together with their parameters. The multiplicity distribution for $R(2,2)$ is a generic distribution with finite support $\set{0,1,5}$. The annotation $( \strBinom, \mathtt{10}, \mathtt{1/3} )$ of $R(1,2)$ in the table specifies that $\#_T R(1,2) \sim \Binom\big(10,\frac13\big)$. That is,
    \[
        \Pr\big( \#_T R(1,2) = k \big) 
        =
        \begin{cases}
            \binom{10}{k} (\frac13)^k (\frac23)^{10-k} & \text{if } 0 \leq k \leq 10\text{ and}\\
            0 & \text{if }k > 10\text.
        \end{cases}
    \]
    The multiplicity probabilities of the other facts are given analogously in terms of the Bernoulli, geometric, and Poisson distributions, respectively. 
    The supports of the multiplicity distributions are $\set{0,1}$ for the Bernoulli, $\set{0,\dots,n}$ for the Binomial, and $\NN$ for both the geometric and Poisson distributions (and finite sets for explicitly encoded distributions). For the first three parameterized distributions, multiplicity probabilities always stay rational if the parameters are rational. This is not the case for the Poisson distribution.
\end{example}

While \cref{def:repsys} seems abstract, this level of detail in the encoding of probability distributions allows us to rigorously discuss computational complexity without resorting to a very narrow framework that only supports some predefined distributions. Our model also comprises tuple-independent set PDBs: The traditional representation system can be recovered from \cref{def:paramtirep} using only the Bernoulli distribution. Moreover, we remark that we can always represent facts that are present with probability $0$, by just omitting them from the tables (for example, fact $R(2,1)$ in \cref{fig:repexample}).

\begin{remark}
    In this work, we focus on TIRS's where the values needed for computation (moments in \cref{sec:evar} and probabilities in \cref{sec:pqe}) are rational. An extension to support irrational values is possible through models of real complexity \cite{Ko1991,BravermanCook2006}. A principled treatment requires a substantial amount of introductory overhead that would go beyond the scope of this paper, and which we therefore leave for future work.
\end{remark}

\section{Expectations and Variances}
\label{sec:evar}

Before computing the probabilities of answer counts, we discuss the computation of the expectation and the variance of the answer count. Recall that in PDBs without multiplicities, the answer to a Boolean query (under set semantics) is either $0$ (i.\,e., $\false$) or $1$ (i.\,e., $\true$). That is, the answer count is a $\set{0,1}$-valued random variable there, meaning that its expectation coincides with the probability of the answer count being $1$. Because of this correspondence, the semantics of Boolean queries on (set) PDBs are sometimes also defined in terms of the expected value \cite{VandenBroeckSuciu2017}. For bag PDBs, the situation is different, and this equivalence no longer holds. Thus, computing expectations, and computing answer count probabilities have to receive a separate treatment. Formally, we discuss the following problems in this section:

\begin{figure}[H]\hfill%
    \begin{subfigure}{.495\textwidth}\centering%
        \paramproblem{$\EXPECTATION_{\rep}( Q )$}
        {A Boolean UCQ $Q$.}                            %
        {A table $T \in \tabs_{\rep}$.}                 %
        {The expectation $\E\big( \#_{T}Q \big)$.}      %
    \end{subfigure}
    \hfill%
    \begin{subfigure}{.495\textwidth}\centering%
        \paramproblem{$\VARIANCE_{\rep}( Q )$}
        {A Boolean UCQ $Q$.}
        {A table $T \in \tabs_{\rep}$.}
        {The variance $\Var\big( \#_T Q  \big)$.}
    \end{subfigure}
    \hfill\mbox{}%
\end{figure}

\subsection{Expected Answer Count}

We have pointed out above that computing expected answer counts for set PDBs and set semantics is equivalent to computing the probability that the query returns $\true$. There are conjunctive queries, for example, $Q = \exists x \exists y \with R(x) \wedge S(x,y) \wedge T(y)$, for which the latter problem is $\sharpP$-hard \cite{Graedel+1998,DalviSuciu2004}. Under a set semantics, disjunctions and existential quantifiers semantically correspond to taking maximums instead of adding multiplicities. Under a bag semantics, we are now able to exploit the linearity of expectation to easily compute expected values, which was not possible under a set semantics. 

\begin{lemma}\label{lem:expectation-of-CQ}
    Let $\D$ be a tuple-independent PDB and let $Q$ be a Boolean CQ, $Q = \exists x_1 \dots \exists x_m \with R_1(\tup t_1) \wedge \dots \wedge R_n( \tup t_n )$. For every $\tup a \in \adom( \D )^{m}$, we let $F( \tup a )$ denote the set of facts appearing in $Q^*[ \tup x / \tup a ]$, and for every $f \in F( \tup a )$, we let $\nu( f, \tup a )$ denote the number of times $f$ appears in $Q^*[ \tup x / \tup a ]$. Then
    \begin{equation}\label{eq:expectation-of-CQ}
        \E\big( \#_{ \D } Q \big) 
            =
        \sum_{ \tup a \in \adom( \D )^m } \prod_{ f \in F( \tup a ) } \E\Big( \big( \#_{\D} f \big)^{ \nu( f, \tup a ) } \Big)\text.
    \end{equation}
\end{lemma}

\begin{proof}
    By definition, we have
    \begin{displaymath}%
        \#_D Q 
            = 
        \sum_{ \tup a \in \adom( D )^m } \#_D( Q^*[ \tup x / \tup a ] )
            =
        \sum_{ \tup a \in \adom( \D )^m } \#_D( Q^*[ \tup x / \tup a ] )
    \end{displaymath}
    for every individual instance $D$ of $\D$. The last equation above holds because, as $Q^*$ is assumed to contain every quantified variable, $\#_D( Q^*[ \tup x / \tup a ] ) = 0$ whenever the tuple $\tup a$ contains an element that is not in the active domain of $D$. By linearity of expectation, we have
    \begin{equation*}
        \E\big( \#_{ \D } Q \big)
            =
        \sum_{ \tup a \in \adom( \D )^m } \E\big( \#_{\D}( Q^*[ \tup x / \tup a ] ) \big)\text.
    \end{equation*}
    Recall, that $Q^*[ \tup x / \tup a ]$ is a conjunction of facts $R_i( \tup t_i[\tup x/\tup a] )$. Thus,
    \begin{math}
        \#_{\D}\big( \bigwedge_{ i = 1 }^{ n } R_i\big( \tup t_i[\tup x/\tup a] \big) \big)
            =
        \prod_{i=1}^{n} \#_{\D}\big( R_i\big( \tup t_i[ \tup x / \tup a ] \big) \big)\text.
    \end{math}
    Because $\D$ is tuple-independent, any two facts in $F( \tup a )$ are either equal, or independent. Therefore,
    \begin{equation*}
        \E\biggl( \prod_{i=1}^{n} \#_{\D} R_i( \tup t_i[ \tup x / \tup a ] ) \biggr)
            =
        \prod_{ f \in F( \tup a ) } \E\big( (\#_{\D}f)^{ \nu( f, \tup a ) } \big)\text,
    \end{equation*}
    as the expectation of a product of independent random variables is the product of their expectations. Together, this yields the expression from \eqref{eq:expectation-of-CQ}.
\end{proof}

By linearity, the expectation of a UCQ is the sum of the expectations of its CQs.

\begin{lemma}\label{lem:expectation-of-UCQ}
    Let $\D$ be a PDB and let $Q = \bigvee_{ i = 1 }^{ N } Q_i$ be a Boolean UCQ. Then we have
    \begin{math}
        \E\big( \#_{ \D } Q \big)
                =
        \sum_{ i = 1 }^{ N } \E\big( \#_{ \D } Q_i \big)\text.
    \end{math}
\end{lemma}

Given that we can compute the necessary moments of fact multiplicities efficiently, \cref{lem:expectation-of-CQ,lem:expectation-of-UCQ} yield a polynomial time procedure to compute the expectation of a UCQ. The order of moments we need is governed by the self-join width of the individual CQs. 
\begin{definition2}
    A TIRS $\rep$ has \emph{polynomially computable moments} up to order $k$, if for all $\lambda \in \Lambda_{\rep}$, we have $\sum_{ n = 0 }^{ \infty } n^{\ell} \cdot P_\lambda(n) < \infty$ and the function $\cod{ \lambda } \mapsto \sum_{ n = 0 }^{ \infty } n^{\ell} \cdot P_\lambda(n)$ can be computed in polynomial time in $\size{ \cod{ \lambda } }$ for all $\ell \leq k$.
\end{definition2}
Before giving the main statement, let us revisit \cref{exa:dists} for illustration.

\begin{example}
    Let $\rep$ be the TIRS from \cref{exa:dists}. The moments of $X \sim \Bern(p)$ are $\E( X^k ) = p$ for all $k \geq 1$. Direct calculation shows that for $X \sim \Binom(n,p)$, the moment $\E( X^k )$ is given by a polynomial in $n$ and $p$. In general, for most of the common distributions, one of the following cases applies. Either, as above, a closed form expression for $\E( X^k )$ is known, or, the moments of $X$ are characterized in terms of the moment generating function (mgf) $\E( e^{tX} )$ of $X$, where $t$ is a real-valued variable. In the latter case, $\E( X^k )$ is obtained by taking the $k$th derivative of the mgf and evaluating it at $t = 0$ \cite[p.~62]{CasellaBerger2002}. An inspection of the mgfs of the geometric, and the Poisson distributions \cite[p.~621f]{CasellaBerger2002} reveals that their $k$th moments are polynomials in their respective parameters as well. Together, $\rep$ has polynomially computable moments up to order $k$ for all $k \in \NN_+$.
\end{example}

\begin{proposition}\label{pro:expectation}
    Let $Q = \bigvee_{i=1}^N Q_i$ be a Boolean UCQ, and let $\rep$ be a TIRS with polynomially computable moments up to order $k$, where $k$ is the maximum self-join width among the $Q_i$. Then $\EXPECTATION_{\rep}\big( \#_T Q \big)$ is computable in polynomial time.
\end{proposition}

\begin{proof}
    We plug \eqref{eq:expectation-of-CQ} into the formula from \cref{lem:expectation-of-UCQ}. This yields at most $\leq N \cdot \size{ \adom( T ) }^m \cdot n$ terms (where $m$ is the maximal number of variables, and $n$ the maximum number of atoms among the CQs $Q_1,\dots,Q_N$). These terms only contain moments of fact multiplicities of order at most $k$.
\end{proof}

We emphasize that the number $k$ from \cref{pro:expectation}, that dictates which moments we need to be able to compute efficiently, comes from the fixed query $Q$ and is therefore constant. More precisely, it is given through the number of self-joins in the query. In particular, if all CQs in $Q$ are self-join free, it suffices to have efficient access to the expectations of the multiplicities.

\subsection{Variance of the Answer Count}

With the ideas from the previous section, we can also compute the variance of query answers in polynomial time. Naturally, to be able to calculate the variance efficiently, we need moments of up to the double order in comparison to the computation of the expected value.

\begin{restatable}{proposition}{varianceproposition}\label{pro:variance}
    Let $Q = \bigvee_{i=1}^N Q_i$ be a Boolean UCQ, and let $\rep$ be a TIRS with polynomially computable moments up to order $2k$, where $k$ is the maximum self-join width among the $Q_i$. Then $\VARIANCE_{\rep}\big( \#_T Q \big)$ is computable in polynomial time.
\end{restatable}

As before, the main idea is to rewrite the variance in terms of the moments of fact multiplicities. This can be achieved by exploiting tuple-independence and linearity of expectation. 
The proof of \cref{pro:variance} can be found in \cref{app:variance}.

Despite the fact that the variance of query answers may be of independent interest, it can be also used to obtain bounds for the probability that the true value of $\#_T Q$ is close to its expectation, using the Chebyshev inequality \cite[Theorem 5.11]{Klenke2014}. This can be used to derive bounds on $\Pr(\#_T Q \leq k)$, when the exact value is hard to compute. 

\begin{remark}
    Proposition \ref{pro:variance} extends naturally to higher-order moments:
    If $\rep$ is a TIRS with polynomially computable moments up to order $\ell \cdot k$ and $Q = \bigvee_{i=1}^N Q_i$ a Boolean UCQ where the maximum self-join width among the $Q_i$ is $k$, then all centralized and all raw moments of order up to $\ell$ of $\#_T Q$ are computable in polynomial time. 
    The proof can be found in \cref{app:higher}.
\end{remark}
\section{Answer Count Probabilities}
\label{sec:pqe}

In this section, we treat the alternative version of probabilistic query evaluation in bag PDBs using answer count probabilities rather than expected values. Formally, we discuss the following problem.

\begin{figure}[H]\centering%
    \begin{minipage}{.55\linewidth}
\paramproblem{$\PQE_{\rep}(Q,k)$}
    {A Boolean (U)CQ $Q$, and $k \in \NN$.}
    {A table $T \in \tabs_{\rep}$.}
    {The probability $\Pr(\#_T Q \leq k)$.}
    \end{minipage}
\end{figure}

This problem amounts to evaluating the cumulative distribution function of the random variable $\#_T Q$ at $k$. The properties of this problem bear a close resemblance to the set version of probabilistic query evaluation, and we hence name this problem \enquote{$\PQE$}.

\begin{remark}
    Instead of asking for $\Pr( \#_T Q \leq k )$, we could similarly define the problem of evaluating the probability that $\#_T Q$ is \emph{at least}, or exactly \emph{equal} to $k$. %
    For Boolean CQs, and the class of representation systems we discuss next, it will turn out that $\PQE_{\rep}(Q,k)$ (in the version with $\leq$) is in polynomial time or $\sharpP$-hard, \emph{independent of $k$}. This directly implies the corresponding statement for $\geq$.\footnote{Except for $k=0$, which is always trivial for $\geq$.} It then also follows that if $\PQE_{\rep}(Q,k)$ is in polynomial time, then the equality version can be solved in polynomial time. However, if $\PQE_{\rep}(Q,k)$ is $\sharpP$-hard for all $k$, we can only immediately infer that there exists $k'$ such that computing the probability of exact answer count $k'$ is $\sharpP$-hard. Indeed, consider a representation system $\rep$ in which every possible multiplicity is even. Then $\Pr( \#_T Q = k' ) = 0$ whenever $k$ is odd, even if $Q$ is hard.

    If $k$ were presented as a \emph{unary} encoded input, then each of the three versions of the problem (\enquote{$\leq k$}, \enquote{$\geq k$}, \enquote{$ = k$}) can be solved in polynomial time using oracle accesses to any of the other versions. In their restricted setting, a corresponding discussion for the case of \emph{binary} encoded $k$ can be found in \cite[Lemma 2]{ReSuciu2009}.
\end{remark}

Throughout this section, calculations involve the probabilities for the multiplicities of individual facts. However, we want to discuss the complexity of $\PQE_{\rep}(Q,k)$ independently of the complexity, in $k$, of evaluating the multiplicity distributions. This motivates the following definition, together with taking $k$ to be a parameter, instead of it being part of the input.

\begin{definition2}\label{def:polyprob}
    A TIRS $\rep$ is called a \emph{p-TIRS}, if for all $k \in \NN$ there exists a polynomial $p_k$ such that for all $\lambda \in \Lambda_{ \rep }$, the function $\cod{ \lambda } \mapsto P_{ \lambda }( k )$ can be evaluated in time $\mathcal O\big( p_k( \size{ \cod{ \lambda } } ) \big)$.
\end{definition2}

\Cref{def:polyprob} captures reasonable assumptions for \enquote{efficient} TIRS's with respect to the evaluation of probabilities: If the requirement from the definition is not given, then $\PQE_{\rep}( \exists x \with R(x), k)$ can not be solved in polynomial time, even on the class of tables that only contain a single annotated fact $R(a)$. This effect only arises due to the presence of unwieldy probability distributions in $\rep$.

As it turns out, solving $\PQE_{\rep}(Q,k)$ proves to be far more intricate compared to the problems of the previous section. For our investigation, we concentrate on self-join free conjunctive queries. While some simple results follow easily from the set semantics version of the problem, the complexity theoretic discussions quickly become quite involved and require the application of a set of interesting non-trivial techniques.

Our main result in this section is a dichotomy for Boolean CQs without self-joins. From now on, we employ nomenclature (like \emph{hierarchical}) that was introduced in \cite{DalviSuciu2007b,DalviSuciu2012}. If $Q$ is a Boolean self-join free CQ, then for every variable $x$, we let $\sg(x)$ denote the set of relation symbols $R$ such that $Q$ contains an $R$-atom that contains $x$. We call $Q$ \emph{hierarchical} if for all distinct $x$ and $y$, whenever $\sg(x) \cap \sg(y) \neq \emptyset$, then $\sg( x ) \subseteq \sg( y )$ or $\sg( y ) \subseteq \sg( x )$. This definition essentially provides the separation between easy and hard Boolean CQs without self-joins. In the bag semantics setting, however, there exists an edge case where the problem gets easy just due to the limited expressive power of the representation system. This edge case is governed only by the probabilities for multiplicity zero that appear in the representations. We denote this set by $\zeroPr(\rep)$, i.e., 
\[\zeroPr( \rep ) = \set{ p \in [0,1] \with P_{\lambda}(0) = p \text{ for some } \lambda \in \Lambda(\rep) }\text.\] 
If a p-TIRS satisfies $\zeroPr(\rep) \subseteq \set{0,1}$, then it can only represent bag PDBs whose deduplication is deterministic. In this case, as we will show in the next subsection, the problem becomes easy even for arbitrary UCQs.\footnote{Using the same definition in the set semantics version of the problem would come down to restricting the input tuple-independent PDB to only use $0$ and $1$ as marginal probabilities, so the problem would collapse to traditional (non-probabilistic) query evaluation. Under a bag semantics, there still exist interesting examples in this class, as the probability distribution over non-zero multiplicities is not restricted in any way.} 

\begin{restatable}{theorem}{ourdichotomy}\label{thm:ourdichotomy}
    Let $Q$ be a Boolean CQ without self-joins and let $\rep$ be a p-TIRS. 
    \begin{enumerate} 
    \item If $Q$ is hierarchical or $\zeroPr(\rep) \subseteq \set{0,1}$, then $\PQE_{\rep}(Q,k)$ is solvable in polynomial time for all $k \in \NN$. 
    \item Otherwise, $\PQE(Q,k)$ is $\sharpP$-hard for all $k \in \NN$.
    \end{enumerate}
\end{restatable}

\begin{remark}
    It is natural to reconsider what happens, if $k$ is treated as part of the input. With a reduction similar to the proof of Proposition 5 in \cite[Proposition 5]{ReSuciu2009}, it is easy to identify situations in which the corresponding problem is $\sharpP$-hard for \emph{binary} encoded $k$. For example, this is already the case for the simple query $\exists x \colon R(x)$,\footnote{The proof in \cite{ReSuciu2009} uses the same query but for a $\Sum$ aggregation over attribute values in their setting.} 
    if the p-TIRS supports all fair coin flips whose outcomes are either a positive integer, or zero.
    A full proof can be found in \cref{app:kinput}.
    Hardness for this simple query does not conflict with the tractability results of \cite{ReSuciu2009,fink2012aggregation}, because they come with strong restrictions to the annotations. These restrictions are violated by the above construction.
\end{remark}

The remainder of this section is dedicated to establishing \cref{thm:ourdichotomy}.

\subsection{Tractable Cases}

Let us first discuss the case of p-TIRS's $\rep$ with $\zeroPr(\rep) \subseteq \set{0,1}$. Here, $\PQE_{\rep}(Q,0)$ is trivial, because the problem essentially reduces to deterministic query evaluation. The following lemma generalizes this to all values of $k$.

\begin{lemma}\label{lem:degenerate}
    If $\rep$ is a p-TIRS with $\zeroPr( \rep ) \subseteq \set{0,1}$, then $\PQE_{\rep}(Q,k)$ is solvable in polynomial time for all Boolean UCQs $Q$, and all $k \in \NN$.
\end{lemma}

\begin{proof}
    Let $\rep$ be any p-TIRS with $\zeroPr(\rep) \subseteq \set{0,1}$ and let $Q$ be any Boolean UCQ. If $P_{\lambda}(0) = 1$ for all $\lambda \in \Lambda_{\rep}$, then $\rep$ can only represent the PDB where the empty instance has probability $1$. In this case, $\#_TQ = 0$ almost surely, so $\Pr( \#_T Q \leq k ) = 1$ for all $k \in \NN$. 
    
    In the general case, suppose $Q = \bigvee_{i=1}^{N} Q_i$ such that $Q_1, \dots, Q_N$ are CQs. Let $A$ be the set of functions $\alpha$ that map the variables of $Q$ into the active domain of the input $T$. We call $\alpha$ \emph{good}, if there exists $i \in \set{1,\dots,N}$ such that all the facts emerging from the atoms of $Q_i$ by replacing every variable $x$ with $\alpha(x)$ have positive multiplicity in $T$ (almost surely). If there are at least $k + 1$ good $\alpha$ in $A$, then $\#_T Q > k$ with probability $1$ and, hence, we return $0$ in this case. Otherwise, when there are at most $k$ good $\alpha$, we restrict $T$ to the set of all facts that can be obtained from atoms of $Q$ by replacing all variables $x$ with $\alpha(x)$ (and retaining the parameters $\lambda$). The resulting table $T'$ contains at most $k$ times the number of atoms in $Q$ many facts, which is independent of the number of facts in $T$. Hence, we can compute $\Pr( \#_T Q \leq k ) = \Pr( \#_{T'} Q \leq k )$ in time polynomial in $T$ by using brute-force.
\end{proof}

From now on, we focus on the structure of queries again. The polynomial time procedure for Boolean CQs without self-joins is reminiscent of the original algorithm for set semantics as described in \cite{DalviSuciu2007b}. Therefore, we need to introduce some more vocabulary from their work. A variable $x$ is called \emph{maximal}, if $\sg( y ) \subseteq \sg( x )$ for all $y$ with $\sg( x ) \cap \sg( y ) \neq \emptyset$. With every CQ $Q$ we associate an undirected graph $G_Q$ whose vertices are the variables appearing in $Q$, and where two variables $x$ and $y$ are adjacent if they appear in a common atom. Let $V_1, \ldots, V_m$ be the vertex sets of the connected components of $G_Q$. We can then write the quantifier-free part $Q^*$ of $Q$ as $Q^* = Q_0^* \wedge \bigwedge_{ i = 1 }^{ m } Q_i^*$ where $Q_0^*$ is the conjunction of the constant atoms of $Q$ and $Q_1^*,\dots,Q_m^*$ are the conjunctions of atoms corresponding to the connected components $V_1,\dots,V_m$. We call $Q_1^*,\dots,Q_m^*$ the \emph{connected components} (short: \emph{components}) of $Q$.

\begin{remark}\label{rem:hierarchicalInComponents}
    If $Q$ is hierarchical, then every component of $Q$ contains a maximal variable.\footnote{This is true, since the sets $\sg(x)$ for the variables of any component have a pairwise non-empty intersection, meaning that they are pairwise comparable with respect to $\subseteq$.} Moreover, if $x$ is maximal in a component $Q_i^*$, then $x$ appears in all atoms of $Q_i^*$.
\end{remark}

\begin{remark}\label{rem:altCQcomponents}
    For every CQ $Q$ with components $Q_1^*, \dots , Q_m^*$, and constant atoms $Q_0^*$, the answer on every instance $D$ is given by the product of the answers of the queries $Q_0,\dots,Q_m$, where $Q_i = \exists \tup x_i \with Q_i^*$ (and $Q_0 = Q_0^*$), and $\tup x_i$ are exactly the variables appearing in the component $Q_i^*$. That is,
    \begin{math}
        \#_D Q 
            = 
        \#_D Q_0^* 
            \cdot 
            \prod_{ i = 1 }^{ m } \#_D\big( \exists \tup x_i \with Q_i^* \big)
            \text.
    \end{math}
    This is shown in \cref{app:comp}.
    If convenient, we therefore use $Q_0 \wedge Q_1 \wedge \dots \wedge Q_m$ as an alternative representation of $Q$.
\end{remark}

The main result of this subsection is the following.

\begin{restatable}{theorem}{PQEhiersjfCQeasy}\label{thm:PQEhiersjfCQeasy}
    Let $\rep$ be a p-TIRS, and let $Q$ be a hierarchical Boolean CQ without self-joins. Then $\PQE_{\rep}(Q,k)$ is solvable in polynomial time for each $k \in \NN$.
\end{restatable}

\begin{sketch}
    The theorem is established by giving a polynomial time algorithm that computes, and adds up the probabilities $\Pr\big( \#_T Q = k' )$ for all $k' \leq k$. The important observation is that (as under set semantics) the components $Q_i$ of the query (and the conjunction $Q_0$ of the constant atom) yield independent events, which follows since $Q$ is self-join free. In order to compute the probability of $\#_T Q = k'$, we can thus sum over all decompositions of $k'$ into a product $k' = k_0 \cdot k_1 \cdot \dots \cdot k_m$, and reduce the problem to the computations of $\Pr\big( \#_T Q_i = k_i )$. Although the cases $k = 0$, and the conjunction $Q_0$ have to be treated slightly different for technical reasons, we can proceed recursively: every component contains a maximal variable, and setting this variable to any constant, the component potentially breaks up into a smaller hierarchical, self-join free CQ. Investigating the expressions shows that the total number of operations on the probabilities of fact probabilities is polynomial in the size of $T$.
\end{sketch}

\begin{remark}
    The full proof of \cref{thm:PQEhiersjfCQeasy} is contained in \cref{app:comp}.
    As pointed out, the proof borrows main ideas from the algorithm for the probabilistic evaluation of hierarchical Boolean self-join free CQs on tuple-independent PDBs with set semantics, as presented in \cite[p.~30:15]{DalviSuciu2012} (originating in \cite{DalviSuciu2004,DalviSuciu2007a}). 
    The extension to multiplicities is essentially the same as in \cite[Lemma 1]{ReSuciu2009} or \cite[Theorem 3]{fink2012aggregation},
    see our discussion of \hyperref[sec:relatedwork]{related work}.
\end{remark}

\subsection{Intractable Cases}\label{ssec:intractable}

We now show that in the remaining case (non-hierarchical queries and p-TIRS's with $\zeroPr(\rep) \cap (0,1) \neq \emptyset$), the problems $\PQE_{\rep}(Q,k)$ are all hard to solve.

Let $Q$ be a fixed query and let $\PQEset(Q)$ denote the traditional set version of the probabilistic query evaluation problem. That is, $\PQEset(Q)$ is the problem to compute the probability that $Q$ evaluates to $\true$ under \emph{set} semantics, on input a tuple-independent \emph{set} PDB. We recall that the bag version $\PQE_{\rep}(Q,k)$ of the problem (introduced at the beginning of the section) takes the additional parameter $k$, and depends on the representation system $\rep$. Let us first discuss $\PQE_{\rep}(Q,k)$ for $k = 0$. In this case, subject to very mild requirements on $\rep$, we can lift $\sharpP$-hardness from the set version \cite{DalviSuciu2012}, even for the full class of UCQs.

\begin{proposition}\label{pro:hardlambdas}
    Let $S \subseteq [0,1]$ be finite and let $\rep$ be a p-TIRS such that $1-p \in \zeroPr(\rep)$ for all $p \in S$. Let $Q$ be a Boolean UCQ. If\/ $\PQEset(Q)$\/ is $\sharpP$-hard on tuple-independent (set) PDBs with marginal probabilities in $S$, then $\PQE_{\rep}(Q,0)$ is $\sharpP$-hard.
\end{proposition}

\begin{proof}
    Let $\D$ be an input to $\PQEset(Q)$ with fact set $F$ where all marginal probabilities are in $S$, given by the list of all facts $f$ with their marginal probability $p_f$.
    For all $p \in S$, pick $\lambda_p \in \Lambda$ such that $P_{\lambda}(0) = 1 - p$.
    Let $T = \bigcup_{f \in F}\set[\big]{ (f,\cod{\lambda_{p_f}})}$, and let $\delta$ be the function that maps every instance $D$ of $T$ to its deduplication $D'$ (which is an instance of $\D$). Then, by the choice of the parameters, we have $\Pr_{ D \sim \sem{T} }\bigl(\delta(D) = D'\bigr) = \Pr_{\D} \bigl(\set{D'}\bigr)$ for all $D'$. 
    Moreover, $\#_D Q > 0$ if and only if $\delta(D)\models Q$. Thus,
    \[
        \Pr_{ D \sim \sem{T} }\bigl( \#_D Q > 0 \bigr)
            =
        \Pr_{ D \sim \sem{T} }\bigl( \delta(D) \models Q \bigr)
            =
        \Pr_{ D' \sim \D }\bigl( D' \models Q \bigr)\text.
    \]
    Therefore, $\PQEset(Q)$ over tuple-independent PDBs with marginal probabilities from $S$ can be solved by solving $\PQE_{\rep}(Q,0)$.
\end{proof}

\begin{remark}
    This is similar to the approach taken in \cite[Theorem 2]{ReSuciu2009} where the authors establish a dichotomy for $\Count$ aggregations over %
    block-independent disjoint set
    PDBs. They do not give the proof details for the hardness part of $\Count$, but it is to be assumed that the main idea is the same.
\end{remark}

Remarkably, \cite[Theorem 2.2]{KenigSuciu2021} show that Boolean UCQs for which $\PQEset(Q)$ is $\sharpP$-hard are already hard when the marginal probabilities are restricted to $S = \set{c, 1}$, for any rational $c\in(0,1)$. Hence, $\PQE_\rep(Q,0)$ is also $\sharpP$-hard on these queries, as soon as $\set{0, 1 - c} \subseteq \zeroPr(\rep)$.

 \smallskip

Our goal is now to show that if $\rep$ is a p-TIRS, then for any Boolean CQ $Q$ without self-joins, $\sharpP$-hardness of $\PQE_{\rep}(Q,0)$ transfers to $\PQE_{\rep}(Q,k)$ for all $k > 0$.

\begin{theorem}\label{thm:zero-to-k-reduction}
    Let\/ $\rep$ be a p-TIRS and let $Q$ be a Boolean self-join free CQ. Then, if\/ $\PQE_{\rep}(Q,0)$ is $\sharpP$-hard, $\PQE_{\rep}(Q,k)$ is $\sharpP$-hard for each $k \in \NN$.
\end{theorem}

Proving \cref{thm:zero-to-k-reduction} is quite involved, and is split over various lemmas in the remainder of this subsection. Let $\rep$ be any fixed p-TIRS and let $Q$ be a Boolean CQ without self-joins. We demonstrate the theorem by presenting an algorithm that solves $\PQE_{\rep}(Q,0)$ in polynomial time, when given an oracle for $\PQE_{\rep}(Q,k)$ for \emph{any} positive $k$. In \cref{app:illustration}, we illustrate this procedure on a concrete example.

Clearly, we cannot simply infer $\Pr( \#_T Q = 0 )$ from $\Pr( \#_T Q \leq k )$. 
Naively, we would want to shift the answer count of $Q$ by $k$, so that the problem could be answered immediately. However, this is not possible in general. Our way out is to use the oracle several times, on manipulated inputs. Since the algorithms we describe are confined to the p-TIRS $\rep$, we are severely restricted in the flexibility of manipulating the probabilities of fact multiplicities: unless further assumptions are made, we can only work with the annotations that are already present in the input $T$ to the problem. We may, however, also drop entries from $T$ or introduce copies of facts using new domain elements.

For a given table $T$ and a fixed single-component query $Q$, we exploit this idea in \cref{alg:inflation} in order to construct a new table $T^{(m)}$, called the \emph{inflation} of $T$ of order $m$. It has the property that $\#_{T^{(m)}} Q$ is the sum of answer counts of $Q$ on $m$ independent copies of $T$. A small example for the result of running \cref{alg:inflation} for $m = 2$ is shown in \Cref{fig:inflatex}. We will later use oracle answers on several inflations in order to interpolate $\Pr(\#_T Q_i = 0)$ per component $Q_i$ individually, and then combine the results together.

\begin{algorithm}
    \small%
    \caption{$\inflate_{Q}(T,m)$}\label{alg:inflation}
    \begin{algorithmic}[1]
        \Parameter Boolean self-join free CQ $Q$ with a single component and no constant atoms
        \Input $T \in \tabs_{\rep}$, $m \in \NN$
        \Output Inflation of order $m$ of $T$: $T^{(m)} = \bigcup_{ i = 1 }^{ m } T_{m,i} \in \tabs_{\rep}$ such that
        \begin{enumerate}[(O1)]
            \item\label{itm:inflation1} for all $i \neq j$ we have $T_{m,i} \cap T_{m,j} = \emptyset$, 
            \item\label{itm:inflation2} for all $i = 1, \dots, m$ we have $\#_{ T_{m,i} } Q \sim \#_T Q$ i.i.d., and
            \item\label{itm:inflation3} $\#_{ T^{(m)} } Q = \sum_{ i = 1 }^{ m } \#_{ T_{m,i} } Q$.
        \end{enumerate}
        \algrule%
        \State Initialize $T_{m,1}, \dots, T_{m,m}$ to be empty.
        \State For each domain element $a$, introduce new pairwise distinct elements $a^{(1)},\dots,a^{(m)}$.%
        \ForAll{relation symbols $R$ appearing in $Q$}
            \State Let $R(t_1,\dots,t_r)$ be the unique atom in $Q$ with relation symbol $R$.
            \ForAll{pairs of the form $\big(R(a_1,\dots,a_r),\cod{\lambda}\big) \in T$}
                    \ForAll{$i = 1, \dots, m$}
                    \State \begin{varwidth}[t]{.9\linewidth}             
                    Add $\big( R(a_{i,1},\dots,a_{i,k}), \cod{\lambda} \big)$ to $T_{m,i}$ where $a_{i,j} = \begin{cases} a_j^{(i)}\text{, if } t_j \text{ is a variable;}\\
                    a_j \text{, if } t_j \text{ is a constant.}
                    \end{cases}$\strut%
                    \end{varwidth}
            \EndFor
        \EndFor
        \EndFor
        \State \Return $T^{(m)} \coloneqq \bigcup_{ i = 1 }^m T_{m,i}$
    \end{algorithmic}
\end{algorithm}\newcommand{\refo}[1]{\hyperref[itm:inflation#1]{(O#1)}}

\begin{restatable}{lemma}{inflateproperties}\label{lem:inflateproperties}
    For every fixed Boolean self-join free CQ $Q$ with a single component and no constant atoms, \cref{alg:inflation} runs in time $\mathcal O\big( \size{ T } \cdot m \big)$, and satisfies the output conditions \textnormal{\refo1}, \textnormal{\refo2} and \textnormal{\refo3}.
\end{restatable}

The proof of \cref{lem:inflateproperties} can be found in \cref{app:pqe}.
The assumption that the input to \cref{alg:inflation} is self-join free with just a single connected component and no constant atoms is essential to establish \refo{1} and \refo{3}, because it eliminates any potential co-dependencies among the individual tables $T_{m,1},\dots,T_{m,m}$ we create. 
The following example shows that this assumption is inevitable, as the conditions of \cref{lem:inflateproperties} can not be established in general.

\begin{example}\label{exa:not-inflatable}
    Let $\rep$ be a TIRS with $\Lambda_{\rep} = \set{ \lambda }$ such that $P_{ \lambda }( 2 ) = P_{ \lambda }( 3 ) = \frac12$ (and $P_{\lambda}( k ) = 0$ for all $k \notin \set{2,3}$). Consider $Q = \exists x \exists y \with R(x) \wedge S(y)$, and $T = \big( (R(1), \cod{\lambda}), (S(1),\cod{\lambda} ) \big) \in \tabs_{\rep}$. Note that $Q$ has two components and, hence, does not satisfy the assumptions of \cref{lem:inflateproperties}. Then $\#_T Q$ takes the values $4$, $6$ and $9$, with probabilities $\frac{1}{4},\frac{1}{2},\frac{1}{4}$. Thus, if $X,Y \sim \#_T Q$ i.i.d., then $X+Y$ is $13$ with probability $\frac{1}{16} + \frac{1}{16} = \frac{1}{8}$. However, for every $T' \in \tabs_{\rep}$, the random variable $\#_{T'} Q$ almost surely takes composite numbers, as it is equal to the sum of all multiplicities of $R$-facts, times the sum of all multiplicities of $S$-facts, both of these numbers being either $0$ or at least $2$. Thus, there exists no $T' \in \tabs_{\rep}$ such that $\#_{T'} Q = X+Y$. 
\end{example}

\begin{figure}%
    \begin{subfigure}[t]{.5\linewidth}\centering%
        {\small\textbf{Table} $T$}\\[1.5ex]%
        \begin{tabular}{ c l }\toprule
            Relation $R$ & Parameter\\\midrule
            $R(a,a,a)$ & $(\strBinom,\mathtt{10},\mathtt{1/3})$\\
            $R(a,b,c)$ & $(\strGeom,\mathtt{1/2})$\\\bottomrule
        \end{tabular}%
    \end{subfigure}%
    \begin{subfigure}[t]{.5\linewidth}\centering%
        {\small\textbf{Table} $T^{(2)}$}\\[1.5ex]%
        \begin{tabular}{ c l }\toprule
            Relation $R$ & Parameter\\\midrule
            $R(a^{(1)},a^{(1)},a)$ & $(\strBinom,\mathtt{10},\mathtt{1/3})$\\
            $R(a^{(1)},b^{(1)},c)$ & $(\strGeom,\mathtt{1/2})$\\
            $R(a^{(2)},a^{(2)},a)$ & $(\strBinom,\mathtt{10},\mathtt{1/3})$\\
            $R(a^{(2)},b^{(2)},c)$ & $(\strGeom,\mathtt{1/2})$\\\bottomrule
        \end{tabular}%
    \end{subfigure}%
    \caption{Example of a table $T$ and its inflation $T^{(2)}$ for the query $\exists x,y \colon R(x,y,a)$.}
    \label{fig:inflatex}
\end{figure}

For this reason, our main algorithm will call \cref{alg:inflation} independently, for each connected component $Q_i$ of $Q$. Then, \cref{alg:inflation} does not inflate the whole table $T$, but only the part $T_i$ corresponding to $Q_i$. 
If $Q = Q' \wedge Q_i$ and 
we denote 
$\#_{T_i } Q_i$ by $X$ and $\#_{T \setminus T_i } Q'$ by $Y$, then replacing $T_i$ in $T$ with its inflation of order $n$ yields a new table with answer count
$(\#_{T \setminus T_i } Q') \cdot (\#_{\smash{T_{\smash{i}}^{(n)}}}Q_i ) = Y \cdot \sum_{i=1}^n X_i$,
where $X_1,\dots,X_n \sim X$ i.i.d.
Before further describing the reduction, we first explore some algebraic properties of the above situation in general.

\begin{restatable}{lemma}{sumofindepcopiestimesrv}\label{lem:sum_of_indep_copies_times_rv}
    Let $X$ and $Y$ be independent random variables with values in $\NN$ and let $k \in \NN$. Suppose $X_1,X_2, \ldots $ are i.i.d.\ random variables with $X_1 \sim X$.  Let $p_0 \coloneqq \Pr(X = 0)$ and $q_0 \coloneqq \Pr(Y = 0)$. Then, there exist $z_1, \ldots, z_k \geq 0$ such that for all $n \in \NN$ we have
    \[
        \Pr\bigg( Y \cdot \sum_{i = 1}^n X_i \leq k\bigg) 
            = 
        q_0 + (1-q_0) \cdot p_0^n + \sum_{j = 1}^k \binom{n}{j} \cdot p_0^{n - j} \cdot z_j \text.
    \]
\end{restatable}

This is demonstrated in \cref{app:pqe}.
We now describe how $p_0 = \Pr( X = 0 )$ can be recovered from the values of $\Pr\big( Y \cdot \sum_{i = 1}^n X_i \leq k\big)$ and $q_0 = \Pr( Y = 0 )$ whenever $q_0 < 1$ and $p_0 > 0$. With the values $z_1, \dots, z_n$ from \cref{lem:sum_of_indep_copies_times_rv}, and $z_0 \coloneqq 1 - q_0$, we define a function
\begin{equation}\label{eq:g}
    g(n) 
        \coloneqq
    \Pr\bigg( Y \cdot \sum_{i = 1}^n X_i \leq k\bigg) - q_0 =
    \sum_{j = 0}^k \binom{n}{j} \cdot p_0^{n-j} \cdot z_j \text.
\end{equation}

Now, for $m \in \NN$ and $x = 0, 1 \ldots, m$, we define
\begin{equation}\label{eq:hm}
    h_m(x)
    \coloneqq g(m+x) \cdot g(m-x)
    = \sum_{j_1,j_2 = 0}^k \binom{m+x}{j_1} \cdot \binom{m-x}{j_2} \cdot p_0^{2m - j_1 - j_2} \cdot  z_{j_1} \cdot z_{j_2}\text.
\end{equation}

Then, for every fixed $m$, $h_m$ is a polynomial in $x$ with domain $\set{0,\ldots,m}$. As it will turn out, the leading coefficient $\lc(h_m)$ of $h_m$ can be used to recover the value of $p_0$ as follows:
Let $j_{\max}$ be the maximum $j$ such that $z_j \neq 0$. Since for fixed $m$, both $\binom{m+x}{j}$ and $\binom{m-x}{j}$ are polynomials of degree $j$ in $x$, the degree of $h_m$ is $2j_{\max}$ and its leading coefficient is 
\[
    \lc(h_m) = (-1)^{j_{\max}} \cdot (j_{\max}!)^{-2}\cdot p_0^{2m - 2j_{\max}} \cdot z_{j_{\max}}^2\text,
\]
which yields
\begin{equation}\label{eq:compute_p0}
    p_0 = \sqrt{\frac{\lc(h_{m+1})}{\lc(h_m)}}\text.
\end{equation}
Thus, it suffices to determine $\lc(h_m)$ and $\lc(h_{m+1})$. However, we neither know $j_{\max}$, nor $z_{j_{\max}}$, and we only have access to the values of $h_m$ and $h_{m+1}$. To find the leading coefficients anyway, we employ the method of finite differences, a standard tool from polynomial interpolation \cite[chapter 4]{Hildebrand1987}. For this, we use the difference operator $\Delta$ that is defined as $\Delta f(x) \coloneqq f(x+1) - f(x)$ for all functions $f$. When $f$ is a (non-zero) polynomial of degree $n$, the difference operator reduces its degree by one and its leading coefficient is multiplied by $n$. Therefore, after taking differences $n$ times, starting from subsequent values of a polynomial $f$, we are left with the constant function $\Delta^n f = n!\cdot \lc(f) \neq 0$. In particular, taking differences more than $n$ times yields the zero function. Hence, we can determine $\lc(f)$ by finding the largest $\ell$ for which $\Delta^\ell f(0) \neq 0$.

\begin{algorithm}
\small%
    \caption{$\solveComponent_{Q}(T, i)$}\label{alg:comp}
    \begin{algorithmic}[1]
        \Parameter Boolean self-join free CQ $Q = Q_0 \wedge \bigwedge_{i=1}^{r} Q_i$ with connected components $Q_1\ldots,Q_r$.
        \Oracle Oracle for $\PQE_{\rep}(Q,k)$ that, on input $\widetilde T$, returns $\Pr(\#_{\widetilde T} Q \leq k)$
        \Input $T \in \tabs_{\rep}$, $i \in \set{1,\ldots, r}$
        \Output $\Pr(\#_{T_i}Q_i = 0)$
        \algrule%
        \iIf{$P_{\lambda}(0) = 1$ for all $\lambda \in \Lambda$} \Return 0 \iEndIf\label{step:ensure-nontrivial-pdbs}
        \State Fix $\lambda$ with $P_{\lambda}(0) < 1$ and suppose $Q = Q' \wedge Q_i$ (cf.\ \cref{rem:altCQcomponents}).
        \If{$Q'$ is empty}
            \State Set $q_0 \coloneqq 0$ and $g(0) \coloneqq 1$.
        \Else
            \State Let $T' \in \tabs_{\rep}$ be the canonical database for $Q'$, with $\lambda_f \coloneqq \lambda$ for all facts $f$ in $T'$.
            \State Calculate $q_0 \coloneqq \Pr\big( \#_{T'}Q' = 0 \big)$ and set $g(0) \coloneqq 1-q_0$.
        \EndIf{}
        \For{$n = 1,2,\ldots, 4k + 1$}
            \State Set $T_i^{(n)} \coloneqq \inflate_{Q_i}(T_i, n)$. 
            \State Set $g(n) \coloneqq \smash{\Pr\big(\#_{T' \cup T_i^{(n)}}Q \leq k\big)} - q_0$, using the oracle.
        \EndFor
        \iIf{$g(k+1) = 0$} \Return 0 \iEndIf \label{step:check-p0-zero}
        \iFor{$x = 0,1,\ldots, 2k$ and $m = 2k, 2k + 1$}
            $h_m(x) \coloneqq g(m+x) \cdot g(m-x)$
        \iEndFor
        \State Initialize $\ell \coloneqq 2k$.\label{step:start-finite-diff}
        \iWhile{$\Delta^{\ell}h_{2k}(0) = 0$} $\ell \coloneqq \ell - 1$ \iEndWhile \label{step:end-finite-diff}
        \State \Return $\sqrt{ \Delta^{\ell}h_{2k+1}(0) / \Delta^{\ell}h_{2k}(0) }$
    \end{algorithmic}
\end{algorithm}

\medskip

The full procedure that uses the above steps to calculate $p_0$ yields \Cref{alg:comp}. Recall that it focuses on a single connected component. To ensure easy access to the value of $q_0$, we utilize a table that encodes the canonical database of the remainder of the query.\footnote{The canonical database belonging to a self-join free CQ is the instance containing the atoms appearing in the query, with all variables being treated as constants.}
Note that $k$ is always treated as a fixed constant, and our goal is to reduce $\PQE_{\rep}(Q,0)$ to $\PQE_{\rep}(Q,k)$.

\begin{lemma}\label{lem:algo2}
    Algorithm~\ref{alg:comp} runs in polynomial time and yields the correct result.
\end{lemma}
\begin{proof}
    With the notation introduced in the algorithm, we let $Y = \#_{T'} Q'$ (or $Y = 1$ if $Q'$ is empty) and $X = \#_T Q_i = \#_{T_i} Q_i$. Then, $q_0=\Pr(Y=0)$ as in \cref{lem:sum_of_indep_copies_times_rv} and the aim of the algorithm is to return $p_0$.
    
    First, line \ref{step:ensure-nontrivial-pdbs} covers the edge case that $\rep$ can only represent the empty database instance. In all other cases, we fix $\lambda$ with $P_{\lambda}(0) > 0$. As $q_0 = 1 - (1 - P_{\lambda}(0))^t$ where $t$ is the number of atoms of $Q'$, we have $q_0 < 1$.
    From \cref{lem:inflateproperties}, we see that $\#_{\smash{T' \cup T_{\smash{i}}^{(n)}}}Q = Y \cdot \sum_{i=1}^n X_i$, so we are in the situation of \cref{lem:sum_of_indep_copies_times_rv}.
    Hence, $g$ and $h_m$ are as in \eqref{eq:g} and \eqref{eq:hm}.
    Now, as $g(k+1) = p_0 \cdot \sum_{j = 0}^k \binom{n}{j} \cdot p_{\smash{0}}^{k-j} \cdot z_j$ with $z_0 = 1-q_0 > 0$, we find that $p_0$ is zero if and only if $g(k+1)$ is zero. This is checked in line \ref{step:check-p0-zero}.
    Finally, the paragraphs following \cref{lem:sum_of_indep_copies_times_rv} apply, and we determine the degree of $h_{2k}$ using the method of finite differences by setting $\ell$ to the maximum possible degree and decreasing it step-by-step as long as $\Delta^\ell h_{2k}(0) = 0$ in lines \ref{step:start-finite-diff} and \ref{step:end-finite-diff}.
    Then, we have $\ell = 2j_{\max}$ and return 
    \[\sqrt{ \frac{\Delta^{\ell}h_{2k+1}(0)}{ \Delta^{\ell}h_{2k}(0)} } = \sqrt{ \frac{ \ell ! \lc(h_{m+1}) }{ \ell! \lc(h_{m})} } \overset{\eqref{eq:compute_p0}}{=}  p_0\text.\]
    
    Concerning the runtime, since $\rep$ is a p-TIRS, all answers of the oracle calls are of polynomial size in the input. Since $k$ is fixed, the algorithm performs a constant number of computation steps and each term in the calculations is either independent of the input or of polynomial size, yielding a polynomial runtime.
\end{proof}

\begin{proof}[Proof of \cref{thm:zero-to-k-reduction}]
    Let $k > 0$ and suppose that we have an oracle for $\PQE_{\rep}(Q,k)$. Let $Q = Q_0 \wedge \bigwedge_{i=1}^{m} Q_i$ be the partition of $Q$ into components, with $Q_0$ being the conjunction of the constant atoms. Then the $\#_T Q_i$ are independent and
    $\#_T Q = \#_T Q_0 \cdot \prod_{i = 1}^m \#_T Q_i$. Therefore,
    \begin{equation*}
        \Pr\big( \#_T Q = 0 \big)  
    = 
    1 - \Pr( \#_T Q_0 \neq 0 ) \cdot \prod_{ i = 1 }^{ m }\big( 1 - \Pr\big( \#_T Q_i = 0 \big) \big)\text.
    \end{equation*}
    As $\Pr( \#_T Q_0 \neq 0 )$ is easy to compute and \cref{alg:comp} computes $\Pr\big( \#_T Q_i = 0 \big)$ for $i=1,\ldots,k$ with oracle calls for $\PQE_{\rep}(Q,k)$, this yields a polynomial time Turing-reduction from $\PQE_{\rep}(Q,0)$ to $\PQE_{\rep}(Q,k)$.
\end{proof}

With the results from the previous subsections, this completes the proof of \cref{thm:ourdichotomy}.

\begin{proof}[Proof of \cref{thm:ourdichotomy}]
For p-TIRS's with $\zeroPr(\rep) \subseteq \set{0,1}$, the statement is given by \cref{lem:degenerate}. By \cref{thm:PQEhiersjfCQeasy}, $\PQE_{\rep}(Q,k)$ is solvable in polynomial time for hierarchical Boolean CQs without self-joins. For the case of $Q$ being non-hierarchical (and $\zeroPr(\rep) \cap (0,1) \neq \emptyset$), let $p \in (0,1)$ such that $p \in \zeroPr(\rep)$. By \cite[Theorem 3.4]{AmarilliKimelfeld2021}, the set version $\PQEset(Q)$ is already hard on the class of tuple-independent set PDBs where all probabilities are equal to $1-p$. It follows from \cref{pro:hardlambdas} that $\PQE_{\rep}(Q,0)$ is $\sharpP$-hard. By \cref{thm:zero-to-k-reduction}, so is $\PQE_{\rep}(Q,k)$ for all $k \in \NN_+$.
\end{proof}

\section{Conclusion}
\label{sec:conclusion}
    We extend the understanding of probabilistic query evaluation towards a model of tuple-independent bag PDBs with potentially infinite multiplicity supports. Our investigations cover the two key manifestations of the problem: computing expectations, and computing the probability of answer counts. While these problems are equivalent for set semantics, their behavior under bag semantics is disparate. On the one hand, we show that expectations, and more generally, moments are easy to compute, even for Boolean UCQs. Computing the probability of answer counts not exceeding some constant $k$ is a more difficult problem, and we obtain a dichotomy between polynomial time and $\sharpP$-hardness that aligns with prior work \cite{DalviSuciu2012,ReSuciu2009,fink2012aggregation,AmarilliKimelfeld2021}.

    There are several open questions of interest, like the complexity of computing answer count probabilities beyond self-join free CQs; the properties of the problem on bag versions of other well-representable classes of PDBs; or a refined analysis of the complexity in terms of $k$ as part of the input or when accessing certain classes of multiplicity distributions. Moreover, some practically relevant multiplicity distributions like the Poisson distribution naturally yield irrational probabilities, which are not covered by our results on answer count probabilities.
\subsection*{Acknowledgments}
The work of Martin Grohe and Christoph Standke has been funded by the German Research Foundation (DFG) under grants GR 1492/16-1 and GRK 2236 (UnRAVeL).

\bibliographystyle{plainurl}

\appendix

\newpage

\section{Proofs Omitted from Section \ref{sec:evar}}

\subsection{Computing the Variance}\label{app:variance}

In this part of the appendix, we prove \cref{pro:variance}.

\varianceproposition*

Recall for (possibly correlated) random variables $X_1,\dots,X_n$, the variance of their sum is equal to
\begin{equation}\label{eq:var-sum}
    \begin{multlined}[b]
        \Var\bigg( \sum_{ i = 1 }^{ n } X_i \bigg )
            =
        \sum_{ i = 1 }^{ n } \Var( X_i ) + \sum_{ i_1 \neq i_2 } \Cov( X_{i_1}, X_{i_2} )\\
            = 
        \sum_{ i = 1 }^{ n } \Big( \E( X_i^2 ) - \E( X_i )^2 \Big)  + \sum_{ i_1 \neq i_2 } \Big(\E\big( X_{i_1} X_{i_2} \big) - \E\big( X_{i_1} \big) \E\big( X_{i_2} \big)\Big)
        \text,
    \end{multlined}
\end{equation}
where the covariance part vanishes if the random variables are independent \cite[Theorem 5.7]{Klenke2014}. Later on, we will need to apply the above for the case where the random variables $X_i$ are products of other random variables. We note that the well-known formula $\Var( X ) = \E( X^2 ) - \E( X )^2$ generalizes to products of independent random variables as follows:
\begin{equation}\label{eq:var-of-prod}
    \begin{aligned}[b]
        \Var\bigg( \prod_{ i = 1 }^{ n } X_i \bigg)
            =
        \E\Bigg( \bigg( \prod_{ i = 1 }^{ n } X_i \bigg)^2 \Bigg) - \E\bigg( \prod_{ i = 1 }^{ n } X_i \bigg)^2
            =
        \prod_{ i = 1 }^{ n } \E\big( X_i^2 \big) - \prod_{ i = 1 }^{ n } \E\big( X_i \big)^2\text.\qedhere
    \end{aligned}
\end{equation}

We continue with an observation concerning the expected value of a product of CQ answers.

\begin{lemma}\label{lem:expectation-prod}
    Let $\D$ be a tuple-independent PDB, and suppose 
    \begin{align*}
        Q_1 &= \exists x_1 \dots \exists x_{m_1} \with R_1( \tup t_1 ) \wedge \dots \wedge R_{n_1}( \tup t_{n_1} )
        \quad\text{and}\\
        Q_2 &= \exists y_1 \dots \exists y_{m_2} \with S_1( \tup u_1 ) \wedge \dots \wedge S_{n_2}( \tup u_{ n_2 } )
    \end{align*}
     are two Boolean CQs. For all\/ $\tup a \in \adom( \D )^{m_1}$ and\/ $\tup b \in \adom( \D )^{ m_2 }$, we let $F_1( \tup a )$ and $F_2( \tup b )$ be the sets of facts appearing in $Q_1^*[ \tup x / \tup a ]$ and $Q_2^*[ \tup y / \tup b ]$, respectively. Moreover, for every $f \in F_1( \tup a ) \cup F_2( \tup b )$, we let $\nu_1( f, \tup a )$ and $\nu_2( f, \tup b )$ denote the number of times $f$ appears in $Q_1^*[\tup x / \tup a ]$, and in $Q_2^*[\tup y/ \tup b]$. Then
    \[
        \E\big( \#_{\D} Q_1 \cdot \#_{\D} Q_2 \big)
            =
        \sum_{ \tup a,\tup b } 
            \prod_{ f \in F_1(\tup a ) \cup F_2( \tup b) } 
            \E\Big( \big(\#_{\D} f\big)^{ \nu_1( f, \tup a ) + \nu_2( f, \tup b ) }\Big)\text,
    \]
    where\/ $\tup a \in \adom( \D )^{ m_1 }$ and\/ $\tup b \in \adom( \D )^{ m_2 }$.
\end{lemma}

\begin{proof}
    Recall the observation from the equation at the beginning of the proof of \cref{lem:expectation-of-CQ}. Unfolding the definition of $\#_{\D}$, we get
    \[
    \begin{multlined}[t]
        \E\big( \#_{\D} Q_1 \cdot \#_{\D} Q_2 \big) 
        \\
            =
        \E
            \bigg(
            \Big( \sum_{ \tup a \in \adom( \D )^{ m_1 } } \prod_{ f \in F_1( \tup a ) } \big( \#_{\D} f \big)^{ \nu_1( f, \tup a) } \Big)
            \cdot 
            \Big( \sum_{ \tup b \in \adom( \D )^{ m_2 } } \prod_{ f \in F_2( \tup b ) } \big( \#_{\D} f \big)^{ \nu_2( f, \tup b) } \Big) 
            \bigg)\text.
            \end{multlined}
    \]
    Expanding the product, this is equal to 
    \[
        \E\bigg( \sum_{ \tup a, \tup b } \prod_{ f \in F_1( \tup a ) \cup F_2( \tup b ) } \big( \#_{\D} f \big)^{ \nu_1(f,\tup a)+ \nu_2(f, \tup b) }\bigg)\text.
    \]
    The claim then follows from the linearity of expectation, and using that the terms $\#_{\D} f$ are independent for different $f$.
\end{proof}

For obtaining a polynomial time algorithm for the variance, we proceed similar to the previous section. In particular, we start with expressing the variance of a CQ in terms of moments of fact multiplicities.

\begin{lemma}\label{lem:variance-of-CQ}
    Let $\D$ be a tuple-independent PDB and let $Q$ be a Boolean CQ, $Q = \exists x_1 \dots \exists x_m \with R_1( \tup t_1 ) \wedge \dots \wedge R_n( \tup t_n )$.
    For all facts $f$, and all\/ $\tup a \in \adom( \D )^m$, we let $X_{f,\tup a}$ denote the random variable $(\#_{\D} f)^{\nu( f, \tup a ) }$. Then
    \begin{equation}\label{eq:variance-of-CQ}
        \Var\big( \#_{\D} Q \big)
            = 
        \begin{multlined}[t]
            \sum_{ \tup a }  
                \prod_{ f \in F( \tup a ) }  \E\big( X_{f,\tup a}^2 \big) 
                - \prod_{ f \in F( \tup a ) } \E\big( X_{f,\tup a} \big)^2\\
                + \sum_{ \tup a_1 \neq \tup a_2 } \bigg( \prod_{ f \in F(\tup a_1)\cup F( \tup a_2 ) } \E\big(X_{f,\tup a_1} X_{f,\tup a_2} \big)
                    - \smashoperator{\prod_{ f \in F( \tup a_1 ) } } \E\big( X_{f,\tup a_1} \big) 
                    \smashoperator{ \prod_{ f \in F( \tup a_2 ) } } \E\big( X_{f,\tup a_2} \big) \bigg)\text.
        \end{multlined}
    \end{equation}
\end{lemma}

\begin{proof}
    This is obtained by the direct calculation, starting from \eqref{eq:var-sum}.
    For the left part of \eqref{eq:var-sum} (the individual variances), we use \eqref{eq:var-of-prod} and the observations from the proof of \cref{lem:expectation-of-CQ}.
    For the right part of \eqref{eq:var-sum} (the covariances), we use \cref{lem:expectation-prod} and \cref{lem:expectation-of-CQ}.
\end{proof}

Another direct application of \eqref{eq:var-sum} yields the following.

\begin{lemma}\label{lem:variance-of-UCQ}
    Let $\D$ be a PDB and let $Q = \bigvee_{ i = 1 }^{ N } Q_i$ be a Boolean UCQ. Then
    \begin{equation}\label{eq:variance-of-UCQ}
        \Var\big( \#_{\D} Q \big)
            =
        \sum_{ i = 1 }^{ N } 
            \Var( \#_{\D} Q_i ) +
        \sum_{ i_1 \neq i_2 } \Big(
            \E\big( \#_{\D} Q_{i_1} \cdot \#_{\D} Q_{i_2} \big)
            - \E\big( \#_{\D} Q_{i_1} \big) \E\big( \#_{\D} Q_{i_2} \big) \Big)
    \end{equation}
\end{lemma}

Again, if we can compute the necessary moments of fact multiplicities efficiently, then \cref{lem:variance-of-UCQ} and \cref{lem:expectation-of-CQ} yield a polynomial time algorithm to compute the variance of a UCQ.

\begin{proof}[Proof of \cref{pro:variance}]
    Plugging \eqref{eq:variance-of-CQ} into \eqref{eq:variance-of-UCQ}, we obtain a formula with a polynomial number of terms in the input size. The individual terms are moments of fact multiplicities, possibly up to order at most $2k$.
\end{proof}

\subsection{Higher-Order Raw and Central Moments}\label{app:higher}

The ideas for computing the expectation and variance of query answer counts can be extended to higher raw and central moments as follows.

\begin{lemma}\label{lem:expectation-prod-new}
    Let $\D$ be a tuple-independent PDB and let $Q_1,\dots,Q_{\ell}$ be Boolean CQs of the shape 
    \[
        Q_i = \exists x_1^i, \dots, x_{m_i}^i \with R_1^i( \tup t_1^i ) \wedge \dots \wedge R_{n_i}^i( \tup t_{n_i}^i )\text.
    \]
    For all $\tup a^i \in \adom( \D )^{ m_i }$, $i = 1,\dots,\ell$, let $F^i(\tup a^i)$ be the set of facts appearing in $Q_i^*[ \tup x^i / \tup a^i ]$ where $\tup x^i = (x_1^i,\dots,x_{m_i}^i)$, and for all $f \in F^i( \tup a^i )$ let $\nu_i( f, \tup a^i )$ denote the number of times $f$ appears in $Q_i^*[ \tup x^i / \tup a^i ]$. Then
    \[
        \E\big( \#_{\D} Q_1 \cdot \dots \cdot \#_{\D} Q_{\ell} \big) 
        = \sum_{ \tup a^1 \in \adom( \D )^{m_1} } \dots \sum_{ \tup a^{\ell} \in \adom( \D )^{ m_{\ell} } } 
            \prod_{ f \in F^1( \tup a^1 ) \cup \dots \cup F^{\ell}( \tup a^\ell ) }
            \E\Big( (\#_{\D}f)^{ \nu_1( f, \tup a^1 ) + \dots + \nu_{\ell}( f, \tup a^{\ell} ) } \Big)
            \text.
    \]
    The number of expected value terms appearing on the right-hand side above is polynomial in the size of the active domain of $\D$.
\end{lemma}

\begin{proof}
    By definition of $\#_{\D}Q$ for Boolean CQs $Q$, we get
    \begin{align*}
        \E\big( \#_{\D} Q_1 \cdot \dots \cdot \#_{\D} Q_{\ell} \big)
        &= \E\bigg( \prod_{ i = 1 }^{ \ell } \Big(
            \sum_{ \tup a^i \in \adom( \D )^{ m_i } } 
            \prod_{ f \in F^i( \tup a^i ) } (\#_{\D}f)^{ \nu_i(f, \tup a^i) } 
        \Big) \bigg)\\
        &= \E\bigg( \sum_{ \tup a^1 \in \adom( \D )^{m_1} } \dots \sum_{ \tup a^\ell \in \adom( \D )^{ m_{\ell} } } 
            \prod_{ f \in F^1( \tup a^1 ) \cup \dots \cup F^{\ell}( \tup a^{\ell} ) } 
            ( \#_{\D}f )^{ \sum_{ i = 1 }^\ell \nu_i( f, \tup a^i ) } 
            \bigg)\text.
    \end{align*}
    The claim then follows from the linearity of expectation and using that the random variables $\#_{\D}f$ are independent for different $f$.
\end{proof}

\begin{remark}
    We treat queries as functions from the possible world of a probabilistic database $\D$ to $\NN$. If $Q = Q'$, then $Q$ and $Q'$ are the same function, defined on the same probability space $\D$, i.e., $\#_{\D}Q$ and $\#_{\D}Q'$ are the \emph{same random variable}. In our proofs, we manipulate the functional definitions of the query semantics directly. In particular, \cref{lem:expectation-prod-new} applies to the case $Q_1 = \dotsc = Q_\ell = Q$, where it provides an expression for $\E\big( (\#_{\D}Q)^{\ell} \big)$.
\end{remark}

Recall that by linearity
\begin{align}
    \E\Bigl( \bigl( X_1 + \dots + X_n \bigr)^{\ell} \Bigr) 
    &= \sum_{ i_1,\dots,i_{\ell} \in \set{1,\dots,n} } \E\bigl(  X_{i_1}\cdot \dots \cdot X_{i_{\ell}} \bigr)\label{eq:sumpower}\text,\\
    \E\Bigl( \bigl( X - \E(X) \bigr)^{\ell} \Bigr)
    &= \sum_{ i = 0 }^{ \ell } \binom{\ell}{i} \E\Bigl( X^{i} \cdot \bigl( -\E(X) \bigr)^{\ell-i} \Bigr)
    = \sum_{ i = 0 }^{ \ell } \binom{\ell}{i} \cdot \bigl( -\E(X) \bigr)^{\ell-i} \cdot \E( X^{i} )
    \label{eq:centralpower}\text.
\end{align}
for random variables $X, X_1,\dots,X_n$ (defined on the same probability space) and $\ell \in \NN$.

\begin{corollary}
    Let $Q = Q_1 \vee Q_2 \vee \dots \vee Q_N$ be a Boolean UCQ such that $k$ is the maximum self-join width among $Q_1,\dots,Q_N$. Let $\ell \in \NN$. If\/ $\rep$ is a TIRS with polynomially computable moments up to order $\ell \cdot k$, then $\E( (\#_{\D}Q)^{\ell} )$ and $\E\bigl( \bigl(\#_{\D}Q - \E( \#_{\D}Q ) \bigr)^{\ell} \bigr)$ is computable in polynomial time.
\end{corollary}

\begin{proof}
    Using \eqref{eq:sumpower}, we have 
    \[
        \E\bigl( (\#_{\D}Q)^\ell \bigr)
        = \sum_{ i_1,\dots,i_{\ell} \in \set{1,\dots,N} } \E\bigl( \#_{\D}Q_{i_1} \cdot \dots \cdot \#_{\D}Q_{i_{\ell}} \bigr)
    \]
    which we can further rewrite using \cref{lem:expectation-prod-new}. This results in a polynomial length sum of products of terms
    \[ \E( \#_{\D}f )^{ \nu_{i_1}(f,a^{i_1}) + \dots + \nu_{i_{\ell}}(f,a^{i_\ell}) } \text,\]
    in which each of the $\ell$ terms in the exponent has value at most $k$. Hence, $\E\bigl( (\#_{\D}Q)^{\ell} \bigr)$ can be computed in polynomial time if $\rep$ has polynomially computable moments up to order $\ell\cdot k$.

    For higher-order central moments, by \eqref{eq:centralpower} we have
    \[
        \E\Bigl( \bigl(\#_{\D}Q - \E(\#_{\D}Q) \bigr)^{\ell} \Bigr) 
        = \sum_{ i = 0 }^{ \ell } \binom{ \ell }{ i } \cdot \bigl( - \E(\#_{\D}Q) \bigr)^{ \ell - i } \cdot \E\bigl( (\#_{\D}Q)^i \bigr) \text.
    \]
    Using the first part of this proof, this expression can be evaluated in polynomial time, if $\rep$ has polynomially computable moments up to order $\ell \cdot k$.
\end{proof}

\section{Proofs Omitted from Section \ref{sec:pqe}}

\subsection{\texorpdfstring{$k$}{k} as Part of the Input}\label{app:kinput}

If $\rep$ is a p-TIRS, then by $\PQE_{\rep}(Q)$ we denote the probabilistic query evaluation problem for $\rep$ and a fixed Boolean UCQ $Q$, where $k$ is treated as a (binary encoded) \emph{input} instead of a parameter. We show that if $\rep$ supports all distributions of the shape $(0 \mapsto 1/2; x \mapsto 1/2)$ for arbitrary positive integers $x$ (see \cref{fig:repexample}), then $\PQE_{\rep}(Q)$ is $\sharpP$-hard for $Q = \exists x \with R(x)$. We do so by using the idea from the proof of \cite[Proposition 5]{ReSuciu2009}, and reduce from $\sharpSUBSETSUM$.

Let $(x_1,\dots,x_n,B)$ be an instance of $\sharpSUBSETSUM$. That is, $x_1,\dots,x_n$ and $B$ are positive integers (given in binary encoding), and we are asked to count the number of subsets $S \subseteq \set{1,\dots,n}$ such that $\sum_{ i \in S } x_i = B$. We construct a tuple-independent PDB $\D$ over a unary relation $R$ with possible (independent) facts $R(1),\dots,R(n)$, where 
\[
    \Pr\bigl( \#_{\D}( R(i) ) \bigr) = 
    \begin{cases}
        x_i & \text{with probability }\frac12\text,\\
        0   & \text{with probability }\frac12\text.
    \end{cases}
\]
Let $D$ be a randomly drawn instance from $\D$. Such $D$ is drawn with probability $\frac1{2^n}$, and we have $\#_D(Q) = B$ if and only if $\sum x_i = B$, where the sum ranges over all $i$ such that $R(i)$ is present in $D$ with positive multiplicity. Hence,
\[
    \Pr\bigl( \#_{D} Q = B \bigr) \cdot 2^n = \size[\Big]{ \set[\Big]{ S \subseteq \set{1,\dots,n} \colon \sum_{ i \in S } x_i = B } }\text.
\]

\subsection{Hierarchical Self-Join Free CQs}\label{app:comp}

The following lemma states how the answer count of a Boolean CQ is given in terms of the answer counts of its individual connected components. Note that this does not require the query to be self-join free.

\begin{lemma}
    Let $Q = \exists \tup x \with Q^*$ be a Boolean CQ with constant atoms $Q_0^*$ and connected components $Q_1^*, \dots, Q_m^*$. Then for all instances $D$ we have
    \[
        \#_D Q
            =
        \#_D Q_0^* \cdot
            \prod_{ i = 1 }^{ m } \#_D\big( \exists \tup x_i \with Q_i^* \big)\text,
    \]
    where for all $i = 1, \dots, m$, the tuple $\tup x_i$ only contains those variables from $\tup x$ that appear in $Q_i^*$.
\end{lemma}

\begin{proof}
    First, note that every variable from $\tup x$ is contained in precisely one of the tuples $\tup x_i$. Thus,
    \begin{align*}
        \#_D Q
            &= 
        \sum_{ \tup a } \#_D Q_0^* \cdot \prod_{ i = 1 }^m \#_D Q_i^*[\tup x/\tup a]
            = 
        \#_D Q_0^* \cdot \sum_{ \tup a_1 } \cdots \sum_{ \tup a_m } \prod_{ i = 1 }^{ m } \#_D Q_i^*[ \tup x_i / \tup a_i ]\\
            &= 
        \#_D Q_0^* \cdot \prod_{ i = 1 }^{ m } \sum_{ \tup a_i } \#_D Q_i^*[ \tup x_i / \tup a_i ]
            =
        \#_D Q_0^* \cdot \prod_{ i = 1 }^{ m } \#_D \big( \exists \tup x_i \with Q_i^* \big)\text,
    \end{align*}
    as claimed.
\end{proof}

In the remainder of this subsection, we develop an algorithm that establishes \cref{thm:PQEhiersjfCQeasy}.

\PQEhiersjfCQeasy*

So let $Q = \exists \tup x \with Q^*$ be a hierarchical, Boolean self-join free CQ, and let $Q_1^*,\dots, Q_m^*$ be its connected components, and let $Q_0^*$ be the conjunction of the constant atoms. 

Let $k \in \NN$. Since 
\[
    \Pr\big( \#_T Q \leq k \big) = \sum_{ k' \leq k } \Pr\big( \#_T Q = k' \big)\text,
\]
it suffices to show that $\Pr\big( \#_T Q = k \big)$ is computable in polynomial time. First, consider the case $k \geq 1$. We have 
\begin{align*}
    \Pr\big( \#_T Q = k \big) 
        &= 
    \Pr\bigg( \#_T Q_0^* \cdot \prod_{ i = 1 }^{ m } \#_T \big( \exists \tup x_i \with Q_i^* \big) = k \bigg)\\
        &= 
    \sum_{ \substack{ k_0, \dots, k_m \\ \prod_i k_i = k } }\mkern-4mu
        \Pr\bigg( \#_T Q_0^* = k_1 \text{ and } \#_T\big( \exists \tup x_i \with Q_i^*\big) = k_i \text{ f.\,a. } i = 1, \dots, m \bigg)
\end{align*}
Since $Q$ is self-join free, the random variables appearing in the above products are independent. Thus,
\begin{equation}\label{eq:PQEsjfCQ-decompose-k}
    \Pr\big( \#_T Q = k \big)
        =
    \sum_{ \substack{ k_0, \dots, k_m \\ \prod_i k_i = k } }
        \Pr\big( \#_T Q_0^* = k_1 \big) \cdot
        \prod_{ i = 1 }^{ m } \Pr\big( \#_T( \exists \tup x_i \with Q_i^* ) = k_i \big)\text.
\end{equation}
Recall, that $Q_0^*$ is the conjunction of the constant atoms of $Q$, say, such that the constant atoms correspond to distinct facts $f_1,\dots,f_{\ell}$ where $f_i$ appears $n_i$ times in $Q_0^*$. Then
\begin{align*}
    \Pr\big( \#_T Q_0^* = k_1 \big)
        &=
    \Pr\bigg( \prod_{ i = 1 }^{ \ell } \#_T f_i = k_1 \bigg)\\
        &=
    \sum_{ \substack{ k_{1,1},\dots,k_{1,\ell} \\ \prod_i k_{1,i} = k_1 } }
        \Pr\big( (\#_T f_i)^{n_i} = k_{1,i} \text{ for all } i = 1,\dots,\ell \big)
\end{align*}
Note that since $k > 0$, it follows that $k_1 > 0$, so the sum appearing above has finitely many terms, and the concrete number only depends on $Q$. Since the random variables $(\#_T f_i)^{n_i}$ are independent, the last probability turns into a product, and the whole expression can be evaluated in polynomial time.

For the other probabilities in \eqref{eq:PQEsjfCQ-decompose-k}, let $\hat{x}_i$ be a maximal variable of $Q_i^*$ and let $\check{\tup x}_i$ denote the tuple $\tup x_i$ with $\hat{x}_i$ omitted. Then, we get
\begin{equation}\label{eq:PQEsjfCQ-decompose-k_i}
    \begin{aligned}[b]
    &\Pr\big( \#_T( \exists \tup x_i \with Q_i^* ) = k_i \big)
        =
    \Pr\bigg( \sum_{ a } \#_T ( \exists \check{\tup x}_i \with Q_i^*[\hat{x}_i/a] ) = k_i \bigg)\\
        &=
    \Pr\Bigg( \bigcup_{ \substack{ (k_a)_{ a \in \adom(T) } \\ \sum_a k_a = k_i} }
        \bigcap_{ a } \mkern8mu
            \#_T ( \exists \check{\tup x}_i \with Q_i^*[\hat{x}_i/a]) = k_a 
        \Bigg)\\
        &=
    \sum_{ \substack{ (k_a)_{ a \in \adom(T) } \\ \sum_a k_a = k_i} }
        \Pr\bigg( 
            \bigcap_{ a } \mkern8mu
                \#_T \big( \exists \check{\tup x}_i \with Q_i^*[\hat{x}_i/a] \big) = k_a
        \bigg)\text.
    \end{aligned}
\end{equation}
Recall that $\hat{x}_i$, as a maximal variable in the connected component $Q_i^*$, appears in every atom of $Q_i^*$. Thus, the facts appearing in $Q_i^*[ \tup x_i / \tup a_i ]$ are disjoint from the ones appearing in $Q_i^*[ \tup x_i / \tup a_i' ]$ whenever $\tup a_i$ and $\tup a_i'$ disagree on the value for $\hat{x}_i$. But this means that the intersection appearing in \eqref{eq:PQEsjfCQ-decompose-k_i} ranges over independent events. Therefore,
\begin{equation}\label{eq:PQEsjfCQ-recursion}
    \Pr\big( \#_T( \exists \tup x_i \with Q_i^* ) = k_i \big)
        =\mkern-10mu
    \sum_{ \substack{ (k_a)_a \\ \sum_a k_a = k_i} }
        \prod_{ a }
        \Pr\big( 
            \#_T ( \exists \check{\tup x}_i \with Q_i^*[\hat{x}_i/a]) = k_a
        \big)\text.
\end{equation}
This can be evaluated in polynomial time, if the probabilities on the right-hand side can be evaluated in polynomial time. For all $a \in \adom(T)$, the Boolean query $\exists \check{x}_i \with Q_i^*[\hat{x}_i/a]$ again decomposes into connected components and possibly some leftover constant atoms. If there are only constant atoms left, then we are done, using the procedure for $Q_0^*$ from above. Otherwise, 
we proceed recursively on the connected components, grounding out a maximal variable in each step. Note that the number of summation terms in \eqref{eq:PQEsjfCQ-recursion} is at most 
\[
    \binom{ \size{ \adom(T) } }{ k_i } \cdot \size{ \adom(T) }^{k_i}\text,
\]
and in each of the summation terms, the product has at most $\size{ \adom(T) }$ factors. That is, the total number of subproblems we created from \eqref{eq:PQEsjfCQ-decompose-k} is polynomial in $\adom(T)$. When turned into a recursive procedure, we remove one quantifier per recursive call. As the query is fixed, this happens a constant number of times. Thus, $\Pr( \#_T Q = k )$ can be computed in polynomial time for $k \geq 1$.

The case $k = 0$ receives a special treatment, because in this case, the sum in \eqref{eq:PQEsjfCQ-decompose-k} would have infinitely many terms. However, $\#_D Q \neq 0$ in any individual instance $D$ if and only if all of $\#_D Q_0^*$ and $\#_D( \exists \tup x_i \with Q_i^* )$ for $i=1,\dots,m$ are different from $0$. Thus, going over the converse probability,
\begin{align}
    \Pr\big( \#_T Q = 0 \big) 
        &= 
    1 - \Pr\big(\#_T Q_0^* \neq 0 \text{ and }\#_T( \exists \tup x_i\with Q_i^*) \neq 0 \text{ f.\,a. }i \big) \notag\\
        &= 
    1 - \Pr( \#_T Q_0 \neq 0 ) \cdot \prod_{ i = 1 }^{ m }\Big( 1 - \Pr\big( \#_T Q_i = 0 \big)\label{eq:zero-prob-of-product} \Big)\text.
\end{align}
As $\Pr( \#_T Q_0^* \neq 0 ) = \prod_{f \in Q_0}(1 - P_{\lambda_f}(0))$, it is computable in polynomial time. For the rest, note that
\[
    \Pr\big( \#_T( \exists \tup x_i \with Q_i^* ) = 0 \big)
        =
    \prod_{ a } \Pr\big( \#_T( \exists \check{\tup x}_i \with Q_i^*[ \hat{x}_i / a ]) = 0 \big)\text,
\]
using again that the random variables $\#_T( \exists \check{\tup x}_i \with Q_i^*[ \hat{x}_i / a ])$ are independent for different $a$. This can be turned into a recursive procedure as before, which again runs in polynomial time. This concludes the proof of \cref{thm:PQEhiersjfCQeasy}.

\subsection{Beyond Hierarchical Queries}\label{app:pqe}

We first analyze the correctness and runtime of \cref{alg:inflation}.

\inflateproperties*

\begin{proof}
    The algorithm iterates over all entries of $T$ and constructs $m$ facts per entry, hence the runtime. Every fact $f$ that results in the construction of facts $f_1,\dots,f_m$ uses a relation symbol that appears in an atom of $Q$. As $Q$ contains no constant atoms, and is self-join free, this atom is unique and contains a variable. The position of this variable has different values in $f_1,\dots,f_m$ according to the construction. Hence, $T_{m,i}$ and $T_{m,j}$ are disjoint in the end for $i \neq j$, showing \refo1. Note that the instances of $T_{m,i}$ are in one-to-one correspondence with the instances from $T$ (after removing from the latter all facts whose relation symbol does not appear in the query) and that $Q$ returns the same answer on two such corresponding instances. Thus, \refo2{} follows. Finally, we obtain \refo3{} by unraveling the definition of $\#_{T^{(m)}}$, and noting that the relevant valuations of the variables of $Q^*$ partition exactly into the sets of domain elements $a_{i,j}$ with $i = 1, \dots, m$. This latter assertion is true because $Q$ has a single connected component.\footnote{Otherwise, valuations in the definition of $\smash[b]{\#_{T^{(m)}}}$ would also include such that use domain elements $a_{i,j}$ for the variables in one connected component, and elements $a_{i',j}$ with $i' \neq i$ for the variables of another component. In this case, such valuations are not counted by $\sum_{ i = 1 }^{ m } \#_{T_{m,i}} Q$.}
\end{proof}

Next, we establish the algebraic properties of the output of the algorithm.

\begin{restatable}{lemma}{sumofindepcopies}\label{lem:sum_of_indep_copies}
    Let $X$ be a random variable with values in $\NN$ and let $k \in \NN$. Suppose $X_1,X_2, \ldots $ are i.i.d random variables with $X_1 \sim X$.  Let $p_0 \coloneqq \Pr(X = 0)$. Then, there exist  $y_1, \ldots, y_k \geq 0$ such that for all $n \in \NN$ we have
    \[
        \Pr\bigg( \sum_{i = 1}^n X_i \leq k\bigg) = p_0^n + \sum_{j = 1}^k \binom{n}{j} \cdot p_0^{n - j} \cdot y_j \text.
    \]
\end{restatable}

\begin{proof}
    To simplify notation, we let $p_\ell \coloneqq \Pr(X = \ell)$ for all $\ell = 1, \dots, k$. Then
    \begin{align*}
        \Pr\bigg( \sum_{i = 1}^n X_i \leq k \bigg) 
        &=  \sum_{\substack{(x_1, \ldots, x_n) \in \NN^n \\ \sum_i x_i \leq k}}
            \Pr\big(X_1 = x_1, \ldots, X_n = x_n\big)\\
       &=  \sum_{\substack{(x_1, \ldots, x_n) \in \NN^n \\ \sum_i x_i \leq k}} 
            \prod_{i = 1}^n p_{x_i}
        =  \sum_{\substack{(x_1, \ldots, x_n) \in \NN^n \\ \sum_i x_i \leq k}} 
            \prod_{\ell = 0}^k p_{\ell}^{\size{\set{i\with x_i = \ell}}}
        \text.
    \end{align*}
    
    Note that the terms of the sum in the last line do not depend on the order of the $x_i$. Instead, it only matters how often every possible value $\ell$ is attained among the $x_1,\dots,x_n$. Denoting this multiplicity by $n_{\ell}$, this information is captured by a tuple $(n_0,\dots,n_k)$ which turns the condition $\sum_{ i = 1 }^{ n } x_i \leq k$ into $\sum_{ \ell = 0 }^{ k } \ell \cdot n_{\ell} \leq k$. As the multinomial coefficient $\binom{n}{n_0,\dots,n_k}$ gives exactly the number of tuples $(x_1,\dots,x_n)$ that correspond to a tuple $(n_0,\dots,n_k)$ as above, we find
    \begin{align*}
        \sum_{\substack{(x_1,\dots,x_n) \in \NN^n \\ \sum_i x_i \leq k}} \prod_{\ell = 0}^k p_\ell^{\size{\set{i\with x_i = \ell}}} 
        &= \sum_{\substack{n_0,\dots, n_k \in \NN \\
                \sum_{\ell} n_\ell = n\\
                \sum_{\ell} \ell \cdot n_\ell \leq k}} 
            \binom{n}{n_0,\dots,n_k} \cdot \prod_{\ell = 0}^k p_\ell^{n_\ell}\\
        &= \sum_{\substack{n_0,\dots, n_k \in \NN \\
            \sum_{\ell} n_\ell = n\\
            \sum_{\ell} \ell \cdot n_\ell \leq k}}
            \binom{n}{n_0} \cdot \binom{n - n_0}{n_0,\dots,n_k} \cdot p_0^{n_0} \cdot \prod_{\ell = 1}^k p_\ell^{n_\ell}
        \text.
    \end{align*} %
    The condition $\sum_{ \ell = 0 }^{ k } \ell \cdot n_\ell \leq k$ enforces that $n_0 \geq n - k$, or $n - n_0 \leq k$. Thus, setting $j \coloneqq n - n_0$, we get
    \begin{equation*}
        \smash[b]{\sum_{\substack{n_0,\ldots, n_k \in \NN \\
            \sum_{\ell} n_\ell = n\\
            \sum_{\ell} \ell n_\ell \leq k}}}
            \binom{n}{n_0} \binom{n - n_0}{n_1,\dots,n_k} p_0^{n_0} \prod_{\ell = 1}^k p_\ell ^{n_\ell}
        = \sum_{j=0}^k p_0^{n-j} \binom{n}{j}  
            \underbrace{\sum_{\substack{n_1,\ldots, n_k \in \NN \\
                \sum_{\ell} n_\ell = j\\
                \sum_{\ell} \ell n_\ell \leq k}} 
                \binom{j}{n_1,\ldots,n_k} \prod_{\ell = 1}^k p_\ell ^{n_\ell}
        }_{ \eqqcolon y_j }\text.
    \end{equation*}
    As $y_0 = \binom{0}{0,\dots,0}\cdot \prod_{\ell=1}^{k} p_{\ell}^0 = 1$, the claim follows.
\end{proof}

Note that in the previous lemma, if $X$ does not take all values from $0$ to $k$ with positive probability, some (or even all) of the $y_j$ may be zero.

As seen in \cref{exa:not-inflatable}, we usually cannot inflate the whole table $T$. Instead, we will only inflate the part $T_i$ corresponding to a single connected component $Q_i$ of $Q = Q' \wedge Q_i$.
The following lemma shows that the structure of this answer count is still quite similar to \cref{lem:sum_of_indep_copies}.

\sumofindepcopiestimesrv*

\begin{proof}
To simplify notation, we let $S \coloneqq \sum_{i = 1}^n X_i$. First, we observe
\[
    \Pr\Big( Y \cdot S \leq k\Big) 
    = 
    \Pr(Y = 0) 
    + \sum_{\ell = 1}^\infty \Pr(Y = \ell) \cdot \Pr\Big(S \leq \big\lfloor \frac{k}{\ell} \big\rfloor\Big)\text.
\] 
From \cref{lem:sum_of_indep_copies}, for every $k \in \NN$ exist non-negative numbers $y_{1,k}, \ldots, y_{k,k}$ such that for all $n \in \NN$ holds
    \[
        \Pr( S \leq k) = p_0^n + \sum_{j = 1}^k \binom{n}{j} p_0^{n - j} y_{j,k} \text.
    \]
Together, this yields
\begin{align*}
&\Pr\Big( Y \cdot S \leq k\Big)
 = 
    q_0 + \sum_{\ell = 1}^\infty q_{\ell} \cdot \Bigg(p_0^n + \sum_{j = 1}^{\lfloor \frac{k}{\ell} \rfloor} \binom{n}{j} p_0^{n - j} y_{j, \lfloor \frac{k}{\ell} \rfloor} \Bigg) \\
    &= q_0 + (1 - q_o) \cdot p_0^n + \sum_{\ell = 1}^k  \sum_{j = 1}^{\lfloor \frac{k}{\ell} \rfloor} q_{\ell} \cdot \binom{n}{j} p_0^{n - j} y_{j, \lfloor \frac{k}{\ell} \rfloor} \\
    &= q_0 + (1 - q_o) \cdot p_0^n + \sum_{j = 1}^{k} \binom{n}{j} p_0^{n - j} \underbrace{
    \sum_{\ell = 1}^{\lfloor \frac{k}{j} \rfloor} 
    q_{\ell} \cdot  y_{j, \lfloor \frac{k}{\ell} \rfloor}}_{\eqqcolon z_j}\text,
\end{align*}
which completes the proof.
\end{proof}

\section{Example for the Reduction from Section~\ref{ssec:intractable}}\label{app:illustration}

In this section, we illustrate the inner workings of the reduction described in \cref{ssec:intractable} with an example. 

\subsection{Setup}
Let $\rep$ be the p-TIRS we use, where we assume that $\zeroPr( \rep ) \cap (0,1) \neq \emptyset$. Consider the following query
\[
    Q = \exists x,y,z \with R(a) \wedge S(z,a) \wedge U(x) \wedge V(x,y) \wedge W(y)\text.
\] 
This query has two components $Q_1 = \exists z\with S(z,a)$ and $Q_2 = \exists x,y \with U(x)\wedge V(x,y) \wedge W(y)$, and a constant part $Q_0 = R(a)$. \emph{Our goal in this section is to solve $\PQE_{\rep}(Q,0)$ in polynomial time by only using an oracle for the problem $\PQE_{\rep}(Q,2)$ (i.e. for $k=2$).} As an example, we work with the following input table $T$ (leaving the concrete parameters unspecified):

\begin{center}\hfill%
    \begin{tabular}[t]{cc}\toprule
         \multicolumn{2}{c}{$R$}\\\midrule
         $a$ & $\lambda_{R(a)}$\\\bottomrule
    \end{tabular}
    \hfill
    \begin{tabular}[t]{cc}\toprule
         \multicolumn{2}{c}{$S$}\\\midrule
         $(a,a)$ & $\lambda_{S(a,a)}$\\
         $(b,a)$ & $\lambda_{S(b,a)}$\\\bottomrule
    \end{tabular}
    \hfill
    \begin{tabular}[t]{cc}\toprule
         \multicolumn{2}{c}{$U$}\\\midrule
         $a$ & $\lambda_{U(a)}$\\
         $b$ & $\lambda_{U(b)}$\\\bottomrule
    \end{tabular}
    \hfill
    \begin{tabular}[t]{cc}\toprule
         \multicolumn{2}{c}{$V$}\\\midrule
         $a,b$ & $\lambda_{V(a,b)}$\\
         $b,b$ & $\lambda_{V(b,b)}$\\\bottomrule
    \end{tabular}
    \hfill
    \begin{tabular}[t]{cc}\toprule
         \multicolumn{2}{c}{$W$}\\\midrule
         $a$ & $\lambda_{W(a)}$\\\bottomrule
    \end{tabular}
    \hfill\mbox{}
\end{center}

Now according to the structure of $Q$, we know that
\[
    \#_T Q = \#_T Q_0 \cdot \#_T Q_1 \cdot \#_T Q_2\text,
\]
and, hence, with using independence, that
\begin{align}
    \Pr( \#_T Q = 0 ) 
    &= \Pr( \#_T Q_0 = 0 \text{ or } \#_T Q_1 = 0 \text{ or } \#_T Q_2 = 0 )\notag{}\\
    &= 1 - \big( 1 - \Pr( \#_T Q_0 = 0 ) \big)\big( 1 - \Pr( \#_T Q_1 = 0 ) \big)\big( 1 - \Pr( \#_T Q_2 = 0 ) \big)\text.\label{eq:apxalgo2ex}
\end{align}

The value $\Pr( \#_T Q_0 = 0 )$ is easy to to determine, as it is just $P_{\lambda_{R(a)}}(0)$. For both of the other two factors, we will have to use \cref{alg:comp}.

\subsection{Execution of Algorithm~\ref{alg:comp}}
We now follow the steps of \cref{alg:comp} for $Q_1 = \exists z \with S(z,a)$. The procedure works analogously for $Q_2$.

Suppose that $\lambda$ is an arbitrary parameter in $\Lambda_{\rep}$ with $P_{\lambda}(0) > 0$. The auxiliary table belonging to the canonical database for the remainder of the query is given by
    \begin{center}\hfill%
        \begin{tabular}[t]{cc}\toprule
             \multicolumn{2}{c}{$R$}\\\midrule
             $a$ & $\lambda$\\\bottomrule
        \end{tabular}
        \hfill
        \begin{tabular}[t]{cc}\toprule
             \multicolumn{2}{c}{$U$}\\\midrule
             $x$ & $\lambda$\\\bottomrule
        \end{tabular}
        \hfill
        \begin{tabular}[t]{cc}\toprule
             \multicolumn{2}{c}{$V$}\\\midrule
             $x,y$ & $\lambda$\\\bottomrule
        \end{tabular}
        \hfill
        \begin{tabular}[t]{cc}\toprule
             \multicolumn{2}{c}{$W$}\\\midrule
             $y$ & $\lambda$\\\bottomrule
        \end{tabular}
        \hfill\mbox{}
    \end{center}

Next, we calculate $q_0 \coloneqq \Pr( \#_{T'} Q' = 0 )$. This is easy, since $\#_{T'} Q' \neq 0$ if and only if all four facts depicted above are present with positive multiplicity. Hence, $q_0 = 1 - \big( 1 - P_{\lambda}(0) \big)^4$. The algorithm then defines $g(0) \coloneqq 1 - q_0$.

In the central for loop, we let $n$ take values from $1$ up to $2k+1 = 5$. For every such $n$, we compute the inflation $T_1^{(n)} \coloneqq \inflate_{Q_1}( T_1, n )$, recalling that $T_1$ is the restriction of $T$ to the facts relevant for $Q_1$ (that is, the $S$-facts). For example, the inflation $T_1^{(3)}$ is given as $T_1^{(3)} = T_{3,1} \uplus T_{3,2} \uplus T_{3,3}$ where $T_{3,i}$ is the following table:
    \begin{center}
        \hfill%
        \begin{tabular}[h]{cc}\toprule
            \multicolumn{2}{c}{$S$}\\\midrule
            $a_i,a$ & $\lambda_{S(a,a)}$\\
            $b_i,a$ & $\lambda_{S(b,a)}$\\\bottomrule
        \end{tabular}
        \hfill\mbox{}
    \end{center}
    Note how the values of the first attributes in the tuples are replaced with the new constants, while the values of the second attribute remain unchanged, depending on which positions in the $S$-atom of $Q_1$ are occupied by variables.
    
    In this situation, the oracle is used to obtain $\Pr\big( \#_{T' \cup T_1^{(n)}}  Q \leq k \big)$, and we set $g( n ) \coloneqq \Pr\big( \#_{T' \cup T_1^{(n)}}  Q \leq k \big) - q_0$ (for all values of $n$ in the aforementioned range $1,2, \dots, 5$).

We know that $\Pr\big( \#_T Q_1 = 0 \big) = 0$ if and only if $g( k+1 ) = g( 3 ) = 0$ (cf. the proof of \cref{lem:algo2}). Thus, we check this special case separately. Before we continue, let us inspect $g(n)$ in more detail. We have
    \begin{align*}
        g(n) & = \Pr\big( \#_{T' \cup T_1^{(n)} } Q \leq 2 \big) - q_0\\
            & = \Pr\big( \#_{T'} Q' \cdot \#_{T_1^{(n)}} Q_1 \leq 2 \big) - q_0\\
            & = -q_0 + \sum_{\substack{j_1,j_2 \in \NN\\j_1\cdot j_2 \leq 2}} \Pr\big( \#_{T'} Q' = j_1 \text{ and } \#_{T_1^{(n)}} Q_1 = j_2 \big)\text.
    \end{align*}
    Writing $q_j \coloneqq \Pr\big( \#_{T'} Q' = j \big)$, we continue
    \begin{align*}
        g(n) &= \begin{multlined}[t]
            -q_0 + \sum_{j_2 \in \NN } q_0 \cdot \Pr\big( \#_{T_1^{(n)}} Q_1 = j_2 \big) + \sum_{j_1 \in \NN_+ } q_{j_1} \cdot \Pr\big( \#_{T_1^{(n)}} Q_1 = 0 \big)\\
            + q_1 \cdot \Pr \big( \#_{T_1^{(n)}} Q_1 = 1 \big)
            + q_1 \cdot \Pr \big( \#_{T_1^{(n)}} Q_1 = 2 \big)
            + q_2 \cdot \Pr \big( \#_{T_1^{(n)}} Q_1 = 1 \big)
            \end{multlined}\\
            &= \begin{multlined}[t]\big(1-q_0\big) \cdot \Pr\big( \#_{T_1^{(n)}} Q_1 = 0 \big)\\
            + q_1 \cdot \Pr \big( \#_{T_1^{(n)}} Q_1 = 1 \big)
            + q_1 \cdot \Pr \big( \#_{T_1^{(n)}} Q_1 = 2 \big)
            + q_2 \cdot \Pr \big( \#_{T_1^{(n)}} Q_1 = 1 \big)
            \end{multlined}
    \end{align*}
    Now, writing $p_j \coloneqq \Pr\big( \#_{T_1} Q_1 = j \big)$, we have
    \begin{align*}
       \Pr\big( \#_{T_1^{(n)}} Q_1 = 0 \big) &= p_0^n\text,\\ 
       \Pr\big( \#_{T_1^{(n)}} Q_1 = 1 \big) &= {\textstyle\binom{n}{1}} p_1 p_0^{(n-1)}\text{, and}\\
       \Pr\big( \#_{T_1^{(n)}} Q_1 = 2 \big) &= {\textstyle\binom{n}{2}} p_1^2 p_0^{(n-2)}\text.
    \end{align*}
    Thus, 
    \begin{align*}
        g(n) &= \big( 1 - q_0 \big) \cdot p_0^n + q_1 \cdot {\textstyle\binom{n}1} \cdot p_1 \cdot p_0^{n-1} + q_1 \cdot \Big( {\textstyle\binom{n}2 \cdot p_1^2 \cdot p_0^{n-1} } \Big) + q_2 \cdot {\textstyle\binom{n}1} \cdot p_1 \cdot p_0^{n-1}\\
            &= p_0^n \cdot {\textstyle\binom{n}0}  \big(1-q_0\big) 
            + p_0^{(n-1)} \cdot {\textstyle\binom{n}1} \big( q_1 p_1 + q_1 p_2 + q_2 p_1 \big) + p_0^{n-2} \cdot {\textstyle\binom{n}2} q_1 p_1^2\text,
    \end{align*}
    which is the shape of $g(n)$ that we know from \eqref{eq:g} (below \cref{lem:sum_of_indep_copies_times_rv}).
    
    We continue by defining 
        \begin{align*}
            h_4(x) &\coloneqq g(4+x) \cdot g(4-x) \quad\text{and}\\
            h_5(x) &\coloneqq g(5+x) \cdot g(5-x)
        \end{align*}
        for all $x = 0,1,2,3,4$ (since $2k = 4$). Using our expression for $g(n)$, we get
        \begin{align*}
            h_4(x) = {}&
                \Big( p_0^{4+x} {\textstyle\binom{4+x}{0}} \big( 1 - q_ 0 \big)
                + p_0^{3+x} {\textstyle\binom{4+x}{1}} \big( q_1p_1 + q_1p_2 + q_2p_1 \big)
                + p_0^{2+x} {\textstyle\binom{4+x}{2}} \big( q_1p_1^2 \big) \Big)\\
                &{} \cdot \Big( p_0^{4-x} {\textstyle\binom{4-x}{0}} \big( 1 - q_ 0 \big)
                + p_0^{3-x} {\textstyle\binom{4-x}{1}} \big( q_1p_1 + q_1p_2 + q_2p_1 \big)
                + p_0^{2-x} {\textstyle\binom{4-x}{2}} \big( q_1p_1^2 \big) \Big)\text.
        \end{align*}
        The function $h_5$ has the same shape, but with $i \pm x$ replaced by $i+1 \pm x$. When multiplying this out, the occurrences of $x$ in exponents vanish. Thus, both $h_4$ and $h_5$ are polynomials in $x$. Moreover, they have the same, even degree.
        
        Let us have a look at two examples. If $p_1$ and $q_1$ are both non-zero, then $h_4$ and $h_5$ have the maximum possible degree $4$. In this case, the leading coefficient $\lc(h_4)$ of $h_4$ is $p_0^4 \cdot (-1)^2 \cdot \big( \tfrac12 \big){}^2 \cdot (q_1p_1^2)^2$ (since $\binom{4+x}{2} = \frac{(4+x)(3+x)}{2}$ and $\binom{4-x}{2} = \frac{(4-x)(3-x)}{2}$) and, similarly, $\lc(h_5) = p_0^6 \cdot (-1)^2 \cdot \big( \frac12 \big){}^2 \cdot ( q_1p_1 )^2$. Here we see already that $\sqrt{ \lc(h_5) / \lc(h_4) } = p_0$, as desired.
    
        If instead we have, say, $q_1 = 0$ but $p_1,q_2 \neq 0$, then both $h_4$ and $h_5$ become polynomials of degree $2$, and their leading coefficients contain $p_0^6$ and $p_0^8$, respectively.
        
        In the next subsection, we illustrate how we computationally find the leading coefficients, and thus, $p_0$, using the finite differences method (cf.\ \cite[Chapter 4]{Hildebrand1987}). After this has been concluded, the value $\Pr\big( \#_T Q_2 = 0 \big)$ is calculated in the exact same fashion, and we can use \eqref{eq:apxalgo2ex} to calculate $\Pr\big( \#_T Q = 0 \big)$.

\subsection{Finding \texorpdfstring{$p_0$}{p\textunderscore0} Using Finite Differences}

Let $h_4$ and $h_5$ be as described above (for component $Q_1$). Recall that $k = 2$. We start with $h_4(0), \dots, h_4(4)$ (because $2k = 4$), and iteratively compute differences, yielding $\Delta h_4$, $\Delta^2 h_4$, $\Delta^3 h_4$ and $\Delta^4 h_4$. Our goal is to find the largest $\ell$ for which $\Delta^{\ell} h_4$ is constant but non-zero.

\begin{figure}[H]
    \centering%
    \begin{subfigure}{.5\textwidth}\raggedright%
    \begin{tikzpicture}
        \matrix (small) [matrix of math nodes,nodes={minimum width=.65cm,minimum height=.5cm,font=\small},nodes in empty cells,column sep=.3cm,row sep=-.1cm]
        {
                                            &[-.3cm] \scriptstyle\color{gray} \Delta^0 h_4
                                                   & \scriptstyle\color{gray} \Delta^1 h_4
                                                        & \scriptstyle\color{gray} \Delta^2 h_4
                                                            & \scriptstyle\color{gray} \Delta^3 h_4
                                                                & \scriptstyle\color{gray} \Delta^4 h_4\\
            \scriptstyle\color{gray} h_4(0) & 144   &     &     &     &     \\
                                            &       & -4  &     &     &     \\
            \scriptstyle\color{gray} h_4(1) & 140   &     & -8  &     &     \\
                                            &       & -12 &     & 0   &     \\
            \scriptstyle\color{gray} h_4(2) & 128   &     & -8  &     & 0   \\
                                            &       & -20 &     & 0   &     \\
            \scriptstyle\color{gray} h_4(3) & 108   &     & -8  &     &     \\
                                            &       & -28 &     &     &     \\
            \scriptstyle\color{gray} h_4(4) & 80    &     &     &     &     \\
        };
        \begin{scope}[on background layer]
            \node[fit=(small-4-4)(small-8-4),rounded corners,fill=black!8!white] {};
            \node[fit=(small-5-5)(small-7-5),rounded corners,fill=black!5!white] {};
            \foreach \i/\j/\k/\l in {2/2/3/3,4/2/3/3,4/2/5/3,6/2/5/3,6/2/7/3,8/2/7/3,8/2/9/3,10/2/9/3} {
                \draw[lightgray,thick,dotted] (small-\i-\j) to (small-\k-\l);
            }
            \foreach \i/\j/\k/\l in {3/3/4/4,5/3/4/4,5/3/6/4,7/3/6/4,7/3/8/4,9/3/8/4} {
                \draw[lightgray,thick,dotted] (small-\i-\j) to (small-\k-\l);
            }
            \foreach \i/\j/\k/\l in {4/4/5/5,6/4/5/5,6/4/7/5,8/4/7/5} {
                \draw[lightgray,thick,dotted] (small-\i-\j) to (small-\k-\l);
            }
            \foreach \i/\j/\k/\l in {5/5/6/6,7/5/6/6} {
                \draw[lightgray,thick,dotted] (small-\i-\j) to (small-\k-\l);
            }
            \node[fit=(small-4-4),draw,thick,inner sep=0pt] {};
        \end{scope}
    \end{tikzpicture}
    \end{subfigure}%
    \begin{subfigure}{.5\textwidth}\raggedleft%
    \begin{tikzpicture}
        \matrix (small) [matrix of math nodes,nodes={minimum width=.65cm,minimum height=.5cm,font=\small},nodes in empty cells,column sep=.3cm,row sep=-.1cm]
        {
                                            &[-.3cm] \scriptstyle\color{gray} \Delta^0 h_5
                                                   & \scriptstyle\color{gray} \Delta^1 h_5
                                                        & \scriptstyle\color{gray} \Delta^2 h_5
                                                            & \scriptstyle\color{gray} \Delta^3 h_5
                                                                & \scriptstyle\color{gray} \Delta^4 h_5\\
            \scriptstyle\color{gray} h_5(0) & 49    &     &     &     &     \\
                                            &       & -1  &     &     &     \\
            \scriptstyle\color{gray} h_5(1) & 48    &     & -2  &     &     \\
                                            &       & -3  &     & 0   &     \\
            \scriptstyle\color{gray} h_5(2) & 45    &     & -2  &     & 0   \\
                                            &       & -5  &     & 0   &     \\
            \scriptstyle\color{gray} h_5(3) & 40    &     & -2  &     &     \\
                                            &       & -7  &     &     &     \\
            \scriptstyle\color{gray} h_5(4) & 33    &     &     &     &     \\
        };
        \begin{scope}[on background layer]
            \node[fit=(small-4-4)(small-8-4),rounded corners,fill=black!8!white] {};
            \node[fit=(small-5-5)(small-7-5),rounded corners,fill=black!5!white] {};
            \foreach \i/\j/\k/\l in {2/2/3/3,4/2/3/3,4/2/5/3,6/2/5/3,6/2/7/3,8/2/7/3,8/2/9/3,10/2/9/3} {
                \draw[lightgray,thick,dotted] (small-\i-\j) to (small-\k-\l);
            }
            \foreach \i/\j/\k/\l in {3/3/4/4,5/3/4/4,5/3/6/4,7/3/6/4,7/3/8/4,9/3/8/4} {
                \draw[lightgray,thick,dotted] (small-\i-\j) to (small-\k-\l);
            }
            \foreach \i/\j/\k/\l in {4/4/5/5,6/4/5/5,6/4/7/5,8/4/7/5} {
                \draw[lightgray,thick,dotted] (small-\i-\j) to (small-\k-\l);
            }
            \foreach \i/\j/\k/\l in {5/5/6/6,7/5/6/6} {
                \draw[lightgray,thick,dotted] (small-\i-\j) to (small-\k-\l);
            }
            \node[fit=(small-4-4),draw,thick,inner sep=0pt] {};
        \end{scope}
    \end{tikzpicture}
    \end{subfigure}
    \caption{Differences tables for $h_4$ (left) and $h_5$ (right) for the case $(q_0,q_1,q_2) = \big( \tfrac12, 0, \tfrac12 \big)$ and $(p_0,p_1,p_2) = \big( \tfrac12, \tfrac14, \tfrac14\big)$. For easier reading, all entries for are shown as multiples of $16384^{-1} = 2^{-14}$, for example, $h_4(0) = \tfrac{144}{16384}$. The shaded columns indicate the columns in which the values become constant for the first time (darker shade), and $0$ for the first time (lighter shade). The framed cell is the target cell we find in the algorithm. The algorithm returns $\sqrt{(-2)/(-8)} = \tfrac12 = p_0$.}
    \label{fig:findiffsmall}
\end{figure}
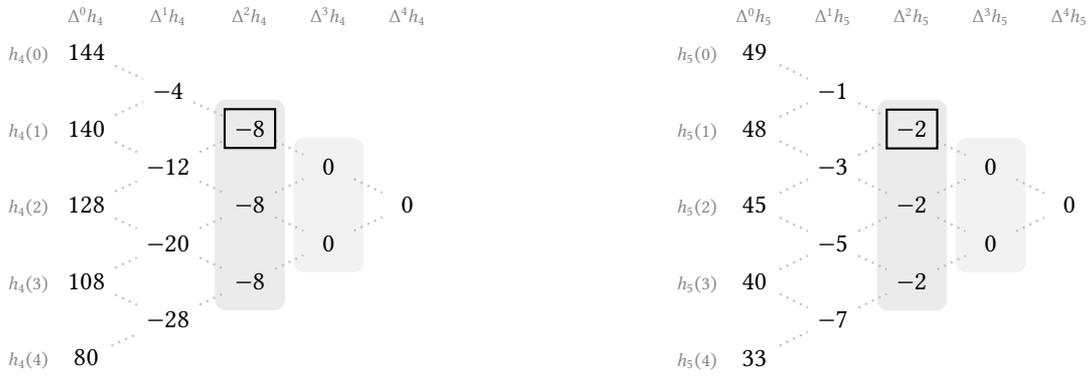

Again, consider two examples. First, suppose that $(q_0,q_1,q_2) = \big( \tfrac12, 0, \tfrac12 \big)$ and $(p_0,p_1,p_2) = \big( \tfrac12, \tfrac14, \tfrac14\big)$ and recall that $p_0$ is actually the unknown value that we want to find. The algorithm only has access to $p_0$ through $h_4$. We have already argued that, $h_4$ has degree $2$ then. Thus, when calculating differences, $\Delta^2 h_4$ will be constant (and non-zero) and $\Delta^3 h_4$ and $\Delta^4 h_4$ will be zero. The constant non-zero value we find for $\Delta^2 h_4$ will be equal to the leading coefficient of $h_4$ times $2!$ by the properties of $\Delta$ (cf.\ \cref{ssec:intractable}). This is illustrated in \cref{fig:findiffsmall}.\footnote{The even degree justifies skipping every second column when looking for the constant non-zero one.}

\begin{figure}[H]
    \centering%
    \begin{subfigure}{.5\textwidth}\raggedright%
    \begin{tikzpicture}
        \matrix (small) [matrix of math nodes,nodes={minimum width=.65cm,minimum height=.5cm,font=\small},nodes in empty cells,column sep=.3cm,row sep=-.1cm]
        {
                                            &[-.3cm] \scriptstyle\color{gray} \Delta^0 h_4
                                                   & \scriptstyle\color{gray} \Delta^1 h_4
                                                        & \scriptstyle\color{gray} \Delta^2 h_4
                                                            & \scriptstyle\color{gray} \Delta^3 h_4
                                                                & \scriptstyle\color{gray} \Delta^4 h_4\\
            \scriptstyle\color{gray} h_4(0) & 2704  &      &       &     &     \\
                                            &       & -120 &       &     &     \\
            \scriptstyle\color{gray} h_4(1) & 2584  &      & -228  &     &     \\
                                            &       & -348 &       & 36  &     \\
            \scriptstyle\color{gray} h_4(2) & 2236  &      & -192  &     & 24  \\
                                            &       & -540 &       & 60  &     \\
            \scriptstyle\color{gray} h_4(3) & 1696  &      & -132  &     &     \\
                                            &       & -672 &       &     &     \\
            \scriptstyle\color{gray} h_4(4) & 1024  &      &       &     &     \\
        };

        \begin{scope}[on background layer]
            \node[fit=(small-6-6),rounded corners,fill=black!8!white] {};
            \foreach \i/\j/\k/\l in {2/2/3/3,4/2/3/3,4/2/5/3,6/2/5/3,6/2/7/3,8/2/7/3,8/2/9/3,10/2/9/3} {
                \draw[lightgray,thick,dotted] (small-\i-\j) to (small-\k-\l);
            }
            \foreach \i/\j/\k/\l in {3/3/4/4,5/3/4/4,5/3/6/4,7/3/6/4,7/3/8/4,9/3/8/4} {
                \draw[lightgray,thick,dotted] (small-\i-\j) to (small-\k-\l);
            }
            \foreach \i/\j/\k/\l in {4/4/5/5,6/4/5/5,6/4/7/5,8/4/7/5} {
                \draw[lightgray,thick,dotted] (small-\i-\j) to (small-\k-\l);
            }
            \foreach \i/\j/\k/\l in {5/5/6/6,7/5/6/6} {
                \draw[lightgray,thick,dotted] (small-\i-\j) to (small-\k-\l);
            }
            \node[fit=(small-6-6),draw,thick,inner sep=0pt] {};
        \end{scope}
    \end{tikzpicture}
    \end{subfigure}%
    \begin{subfigure}{.5\textwidth}\raggedleft%
    \begin{tikzpicture}
        \matrix (small) [matrix of math nodes,nodes={minimum width=.65cm,minimum height=.5cm,font=\small},nodes in empty cells,column sep=.3cm,row sep=-.1cm]
        {
                                            &[-.3cm] \scriptstyle\color{gray} \Delta^0 h_5
                                                   & \scriptstyle\color{gray} \Delta^1 h_5
                                                        & \scriptstyle\color{gray} \Delta^2 h_5
                                                            & \scriptstyle\color{gray} \Delta^3 h_5
                                                                & \scriptstyle\color{gray} \Delta^4 h_5\\
            \scriptstyle\color{gray} h_5(0) & 1156   &      &      &     &     \\
                                            &        & -38  &      &     &     \\
            \scriptstyle\color{gray} h_5(1) & 1118   &      & -73  &     &     \\
                                            &        & -111 &      & 9   &     \\
            \scriptstyle\color{gray} h_5(2) & 1107   &      & -64  &     & 6   \\
                                            &        & -175 &      & 15  &     \\
            \scriptstyle\color{gray} h_5(3) & 832    &      & -49  &     &     \\
                                            &        & -224 &      &     &     \\
            \scriptstyle\color{gray} h_5(4) & 608    &      &      &     &     \\
        };
        \begin{scope}[on background layer]
            \node[fit=(small-6-6),rounded corners,fill=black!8!white] {};
            \foreach \i/\j/\k/\l in {2/2/3/3,4/2/3/3,4/2/5/3,6/2/5/3,6/2/7/3,8/2/7/3,8/2/9/3,10/2/9/3} {
                \draw[lightgray,thick,dotted] (small-\i-\j) to (small-\k-\l);
            }
            \foreach \i/\j/\k/\l in {3/3/4/4,5/3/4/4,5/3/6/4,7/3/6/4,7/3/8/4,9/3/8/4} {
                \draw[lightgray,thick,dotted] (small-\i-\j) to (small-\k-\l);
            }
            \foreach \i/\j/\k/\l in {4/4/5/5,6/4/5/5,6/4/7/5,8/4/7/5} {
                \draw[lightgray,thick,dotted] (small-\i-\j) to (small-\k-\l);
            }
            \foreach \i/\j/\k/\l in {5/5/6/6,7/5/6/6} {
                \draw[lightgray,thick,dotted] (small-\i-\j) to (small-\k-\l);
            }
            \node[fit=(small-6-6),draw,thick,inner sep=0pt] {};
        \end{scope}
    \end{tikzpicture}
    \end{subfigure}
    \caption{Differences tables for $h_4$ (left) and $h_5$ (right) for the case $(q_0,q_1,q_2) = \big( \tfrac12, \tfrac12, 0 \big)$ and $(p_0,p_1,p_2) = \big( \tfrac12, \tfrac14, \tfrac14\big)$. For easier reading, all entries for are shown as multiples of $65536^{-1} = 2^{-16}$, for example, $h_4(0) = \tfrac{2704}{65536}$. The algorithm returns $\sqrt{6/24} = \tfrac12 = p_0$.}
    \label{fig:findiffbig}
\end{figure}
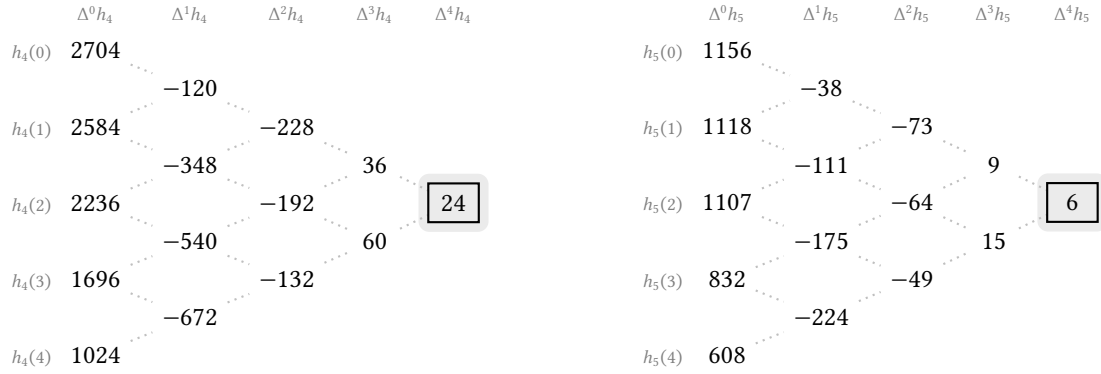

Our second example (see \cref{fig:findiffbig}) illustrates the setup $(q_0,q_1,q_2) = \big( \tfrac12, \tfrac12, 0 \big)$ and $(p_0,p_1,p_2) = \big( \tfrac12, \tfrac14, \tfrac14\big)$. This is one of the cases where $h_4$ and $h_5$ have the maximum possible degree $4$. Here (when seen from the end), there are no zero columns, and the value we are looking for is contained in the very last cell.

\clearpage

\footnotesize

\end{document}